\documentclass[aps,twocolumn]{revtex4}  
\usepackage{graphicx}  
\usepackage{dcolumn}   
\usepackage{bm}        
\usepackage{amssymb}   
\usepackage{subeqnarray}
\usepackage{amsmath}
\usepackage[hypcap]{caption}
\usepackage{cases}
\usepackage{multirow}
\usepackage{amsthm}
\usepackage{bm}
\usepackage{caption}
\usepackage{graphicx}
\usepackage{subfigure}

\newtheorem{thm}{Theorem}
\newtheorem{prop}[thm]{Proposition}

\DeclareMathOperator{\Img}{Im}
\DeclareMathOperator{\Rea}{Re}
\DeclareMathOperator{\Deg}{deg}
\DeclareMathOperator{\Ind}{Index}

\begin{document}



\title{Artificial Gauge Fields and Spin-Orbit Couplings in Cold Atom Systems}
\author{Zhang, Junyi} \affiliation{D\'epartement de Physique, l'\'Ecole Normale sup\'erieure, 24 rue d'Ulm, 75005, Paris, France}
\date{\today}

\begin{abstract}
  This article is a report of  Projet bibliographique  of M1 at \'Ecole Normale Sup\'erieure.  In this article we reviewed the historical developments in artificial gauge fields and spin-orbit couplings in cold atom systems.  We resorted to origins of literatures to trace the ideas of the developments. For pedagogical purposes, we tried to work out examples carefully and clearly, to verified the validity of various approximations and arguments in detail, and to give clear physical and mathematical pictures of the problems that we discussed.  The first part of this article introduced the fundamental concepts of Berry phase and Jaynes-Cummings model.  The second part reviewed two schemes to generate artificial gauge fields with N-pod scheme in cold atom systems.  The first one is based on dressed-atom picture which provide a method to generate non-Abelian gauge fields with dark states.  The second one is about rotating scheme which is achieved earlier historically. Non-Abelian gauge field inevitably leads to spin-orbit coupling.  We reviewed some developments in achieve spin-orbital coupling theoretically and experimentally.  The fourth part was devoted to recently developed idea of optical flux lattice that provides a possibility to reach the strongly correlated regime in cold atom systems.  We developed a geometrical interpretation based on Cooper's theory. Some useful formulae and their proofs were listed in the Appendix.
\end{abstract}
\maketitle

\tableofcontents

\section{Introduction}
\subsection{Berry Phase}
The dynamics of quantum physics was represented by unitary operators.  All the observables are connected to their expectation values of the corresponding operator sandwiched by state vectors and their dual.  In most cases, physical realities depend only on the modular square of the wave function (by Born's probabilistic interpretation~\cite{Born1954}) rather than the wave function itself.  Therefore it is not important if the wave function is multiplied by some phase factor.  Nonetheless, sometimes the results do depend on the phase factors.  One used to argue that those results depending on the phase factor are not gauge invariant, thus they are not observable directly, and these factors can be gauged out by a proper gauge transform.  Therefore, the effects of the phase factor has been long neglected.  While historically, many effects, as A-B effect~\cite{ABeffectES}~\cite{ABeffect}~\cite{ABExpTonomura1}~\cite{ABExpTonomura2}, and nuclear motion in molecules~\cite{Mead1979}~\cite{MeadRevModPhys} were predicted and observed which are clear clues that quantum phase does play some ``observable'' roles in physical reality.  It was Berry~\cite{Berry1984} who pointed out the importance of phase factor in a adiabatic circular evolution~\cite{BornFock}~\cite{Messiah}~\cite{BOApprox} of a quantum state, which is now often called Berry phase. (It is also called Mead-Berry Phase. One also uses geometric phase as synonym.)

Simon~\cite{Simon1983} pointed out the mathematical structure of the Berry phase. He attributed Berry's idea to the underlying holonomical structure of vector bundle.  This idea stimulated understanding the topological characters of quantum Hall effect~\cite{TKNN}, which developed to another exciting field in condensed matter physics.  Wilczek and Zee~\cite{WilczekZee} generalized the ideas of Berry and Simon to non-Abelian gauge fields and proposed possible methods to observe these effects~\cite{ZeeQuadrupole}.

In cold atom systems, as the motion of the atom can be considered adiabatic, then a nontrivial phase factor is induced by light-atom couplings which were well known in quantum optics.  Here we shall follow ideas of Berry, Simon, and Wilczek and Zee to develop the formulations of gauge structure induced by adiabatic evolution.  The coupling between light and beam will be analyzed in the following section by introducing Jaynes-Cumming's model.  The internal degrees of freedom play a role of pseudo-spin.  The  induced  artificial non-Abelian gauge fields will naturally induce the spin-orbital coupling in a system, which has important and interesting features.

In Ref.~\citenum{Berry1984}, Berry first considered a non-degenerate state evolving adiabatically under a time dependent Hamiltonian that is parameterized by a circuit.  In addition to the dynamical phase factor $exp(-\frac{i}{/\hbar}\int_0^t E(\mathbf{R}(t'))dt' )$, the state acquires another phase factor $exp(i\gamma(t))$, where $\dot{\gamma}(t)=i\langle \psi(\mathbf{R}(t))| \nabla_\mathbf{R}|\psi(\mathbf{R}(t))\rangle\cdot \dot{\mathbf{R}}(t)$, $\mathbf{R}$ is the parameter.   If the evolution path is not closed, this phase factor can be gauged away.  While if the Hamiltonian gets back to its initial value, then this additional phase accumulated along the circuit $C$ is
\begin{equation}\label{eq:BerryPhase}
\begin{split}
  \gamma(C)=&i\oint_C \langle\psi(\mathbf{R}(t))| \nabla_\mathbf{R}|\psi(\mathbf{R}(t))\rangle\cdot \dot{\mathbf{R}}(t) dt\\
  =&i\oint_C\langle\psi(\mathbf{R})| \nabla_\mathbf{R}|\psi(\mathbf{R})\rangle\cdot d{\mathbf{R}},
\end{split}
\end{equation}
which is time independent under adiabatic approximation. Since $\langle\psi| \big(\nabla|\psi \rangle \big)=- \big(\nabla\langle\psi| \big) |\psi \rangle $, $\gamma$ is real.  By Stokes theorem, we have
 \begin{equation}\label{eq:BerryPhase2}
\begin{split}
  \gamma(C)
  =&-\Img\oint_C\langle\psi(\mathbf{R})| \nabla_\mathbf{R}|\psi(\mathbf{R})\rangle\cdot d{\mathbf{R}}\\
  =&-\Img\iint_{S(\partial S=C)}d\mathbf{S}\cdot\nabla \times \langle\psi(\mathbf{R})| \nabla_\mathbf{R}|\psi(\mathbf{R})\rangle\\
  =&-\Img\iint_{S(\partial S=C)}d\mathbf{S}\cdot\nabla \times \mathbf{A},
\end{split}
\end{equation}
which is gauge invariant.  Simon~\cite{Simon1983} point out that the integrand is in fact a two form relating to Chern class, which characterize the topology of space of $\mathbf{R}$.   It also indicates the differential structure of the bundle~\cite{Simon1983}~\cite{WuYang}.

In analog to the gauge structure of electromagnetic field,  Yang and Mills generalized $U(1)$ gauge to non-Abelian case~\cite{YangMills}, which is fundamental to the standard model of particle physics.  Wu and Yang also studied nonintegrable phase factor in both Abelian and non-Abelian cases~\cite{WuYang}.  Similarly, Berry and Simon's idea was also generalized to non-Abelian case~\cite{WilczekZee}.  Ref.~\citenum{WilczekZee} considered the ``phase factor'' of a degenerate sub-manifold with dimension larger than one, where gauge potential $\mathbf{A}$ becomes a matrix.  In Ref.~\citenum{WilczekZee} they resort to some symmetry of the system to guarantee the degeneracy that is key to the appearance of non-Abelian gauge field.  However, in cold atom system, we can achieve the degenerate submanifold by dark states of atoms interacting with the lights and therefore the non-Abelian gauge fields emerges naturally.

\subsection{Jaynes-Cummings's Model}
In 1960s, Jaynes and Cummings introduced a model to describe the interactions of a two-level atom with cavity modes of electromagnetic field, which is now named after them~\cite{JCModel1963}~\cite{Cummins1965}.  This model proves to be simple but precise enough to describe the actual experiments of Cavity Quantum ElectroDynamics (Cavity QED or CQED).  It also provides a toy model for studying the artificial gauge fields in cold atom systems~\cite{DalibardArtificialGauge}.

According to Janes and Cummings (Ref.~\citenum{JCModel1963}), we may consider a system of a two-level atom and singled mode of a cavity.  The free electromagnetic field subjected to the boundary condition of the cavity can be quantized as
\begin{equation}\label{eq:fieldQuant}
  \mathbf{A}=A\left(\hat{\epsilon} a(t)e^{i\mathbf{p}\cdot\mathbf{x}}+ \hat{\epsilon}^* a^\dagger(t)e^{-i\mathbf{p}\cdot\mathbf{x}}\right),
\end{equation}
where $\hat{\epsilon}$ and $\hat{\epsilon}^*$ are complex polarization of the field, $a$ and $a^\dagger$ are creation and annihilation operators of the photon, and $A$ is the normalization factor depends on the geometry and boundary conditions of the cavity. We denote $|g\rangle$ and $|e\rangle$ for the internal degree of freedom for the two-level atom.  Thus the noninteracting bases are $ |\alpha\rangle \otimes |n\rangle,\ \alpha=e,g,\ n=0,1,2,...$.

The atom coupled to the electromagnetic field through the interaction Hamiltonian of the form
\begin{equation}\label{eq:HamAtomFieldCoupling}
  H_{int}=-\hat{\mathbf{D}}\cdot\mathbf{E},
\end{equation}
where $\hat{\mathbf{D}}$ is the dipole operator acting on the states of the atom
\footnote{Firstly, in Ref.~\citenum{JCModel1963}, they calculated results for both quantized fields and semiclassical approaches. Here we may well start from the most general quantized forms, but most consequences also work in semiclassical scheme, so we omit the $\hat\ $ over the fields, which will not cause any ambiguity from the context. Secondly, this form of coupling is connected to the form of $-\mathbf{j}\cdot\mathbf{A}$ by Fierz-Pauli transformation~\cite{FierzPauli}~\cite{CCTQEDV_2}.}.
For brevity, we may assume that the entries of the dipole operator is of the form $\langle m|\hat{\mathbf{D}}|m'\rangle = \mathbf{d}_0 (1-\delta_{mm'})$ or $\hat{\mathbf{D}}=\mathbf{d}_0(\hat{\sigma}_+ + \hat{\sigma}_-) $. (In general the dipole matrix is of the form
 $\mathbf{D}_{mm'}=\begin{pmatrix} d_{ee} &  d_{eg} \\  d_{ge} &  d_{gg} \end{pmatrix}, d_{eg}=d_{ge}^*$.
)

There are two kinds of couplings: one couples $|g,n+1\rangle$ and $|e,n\rangle$; the other couples $|g,n\rangle$ and $|e,n+1\rangle$.  The former is related to the process of de-excitation of the atom by emitting a photon or excitation by absorbing one photon.  The later is in the contrast.  It seems to be a violation of the energy conservation for the second term but it is not the case as we do not take the motions of the atom in real space in to account~\footnote{Comments in the lectures by Prof. ZHANG Li of IASTU, China and Prof. Jean-Michel Raimond of UPMC, France}.  According to Ref.~\citenum{JCModel1963}, the second kind of coupling is negligible (as the detuning is not far away from the resonance it is equivalent to the seminal Rotating Wave Approximation (RWA)~\cite{HarocheRaimond}). With RWA, the coupling Hamiltonian decoupled into blocks, with the ground state unshifted, $|e,0\rangle$ coupling to $|g,1\rangle$, $|e,1\rangle$ to $|g,2\rangle$ ... Then the Hamiltonian of the $n^{th}$ coupled block is of the form
\begin{equation}\label{eq:HamBlock}
  \begin{pmatrix}
    E_e + (n-1)\hbar\omega  &  \sqrt{n}\hbar \alpha \\
    \sqrt{n}\hbar \alpha^*   &  E_g + n\hbar\omega
  \end{pmatrix}
  =
  \begin{pmatrix}
    E_1    &  V_{12} \\
    V_{21} &  E_2
  \end{pmatrix},
\end{equation}
where $V_{12}=V=V^*_{21}$ is the coulpling of the two quasi-degenerate levels in the $n^{th}$ submanifold (when detuning is not large).  Therefore the system can be solved exactly by diagonalizing the block Hamiltonian.  The new eigenenergies are
\begin{equation}\label{eq:eigEnergyTLS}
  \begin{split}
    E_\pm =& \frac{1}{2}\left((E_1+E_2)\pm \sqrt{{(E_1+E_2)}^2-4(E_1E_2-{|V|}^2)}\right) \\
          =& \frac{(E_1+E_2)}{2}\pm \frac{1}{2}\sqrt{{(E_1-E_2)}^2+4{|V|}^2}\\
          =& \frac{(E_e+E_g)}{2}+(n-\frac{1}{2})\hbar\omega \\
           & \pm \frac{1}{2}\sqrt{{(E_e-E_g-\hbar\omega)}^2+4{n\hbar^2|\alpha|}^2}.
  \end{split}
\end{equation}

We can also obtain the time evolution of the states with this Hamiltonian by solving time-dependent Schr\"odinger equation
\begin{equation}\label{eq:tdschroedinger}
  i\hbar \begin{pmatrix} \dot{a}\\ \dot{b} \end{pmatrix}
  =
  \begin{pmatrix}
    E_1    &  V_{12} \\
    V_{21} &  E_2
  \end{pmatrix}
  \begin{pmatrix} a\\ b \end{pmatrix},
\end{equation}
where $a$ is the amplitude of state $|e,n\rangle$ and $b$ is that of $|g,n+1\rangle$. In Ref.~\citenum{JCModel1963}, they used this model for beam maser, so they considered an atom originally in excited state $|e\rangle$ decaying to the ground state $|g\rangle$, where, specially, n=0 corresponds to the amplitude of spontaneous emission. Eliminate $b$ and its derivative, we obtain
\footnote{Take the derivative of the first line of Eq.~\ref{eq:tdschroedinger} with respect to time; substitute second line for $\dot{b}$, and first line for $b$.}
\begin{equation}\label{eq:eqA}
  -\hbar^2 \ddot{a}=i \hbar(E_1+E_2) \dot{a} -(E_1E_2-{|V|}^2)a.
\end{equation}
Let $a=e^{-i\omega t}$ (test solution), and substitute for Eq.~\ref{eq:eqA}, we have
\begin{equation}\label{eq:eqomega}
  {(\hbar\omega)}^2 - \hbar\omega (E_1+E_2) + (E_1E_2-{|V|}^2) =0.
\end{equation}
It is easy to observe that this is exactly the secular equation we have solved for the eigen-energies, therefore $\hbar\omega_\pm=E_\pm$, and the general solution is $a=A_+ e^{-i\omega_+ t} + A_- e^{-i\omega_- t}$ where $A_+$ and $A_-$ are to be determined by initial conditions.

Since at $t=0$, the atom is in excited state, i.e. $a(t=0)=1,b(t=0)=0$.  Substitute back to Eq.~\ref{eq:tdschroedinger}, we obtain the initial conditions for their first derivatives $i\hbar\dot{a}(t=0)=E_1,i\hbar\dot{b}(t=0)=V^*$. We achieve the solutions
\begin{equation}\label{eq:abAmpl}
\begin{split}
  a=&e^{-i\frac{(E_1+E_2)t}{2\hbar}}\left(\cos\Omega't/2 +i\frac{\Delta}{\hbar\Omega'}\sin\Omega't/2 \right),\\
  b=&e^{-i\frac{(E_1+E_2)t}{2\hbar}}\left(-2i\frac{V^*}{\hbar\Omega'}\sin\Omega't/2 \right),\\
  {|a|}^2=&\cos^2\Omega't/2 +\frac{{\Delta'}^2}{\hbar^2\Omega'^2}\sin^2\Omega't/2 ,\\
  {|b|}^2=&\frac{4{|V|}^2}{\hbar^2\Omega'^2}\sin^2\Omega't/2,\\
\end{split}
\end{equation}
where $\Delta'=E_1-E_2$ is called detuning and $\hbar\Omega'=\sqrt{\Delta'^2+4{|V|}^2}$ is the generalized Rabi frequency.

When $\Delta'\ll\hbar\Omega'$ (about resonance),
\begin{equation}\label{eq:abOscil}
\begin{split}
  {|a|}^2 \approx &\cos^2\Omega't/2,\\
  {|b|}^2 \approx &\sin^2\Omega't/2,\\
\end{split}
\end{equation}
i.e. the atom oscillates between $|e\rangle$ and $|g\rangle$ of frequency $\Omega'$.
When $V \rightarrow 0$, take the limit of $t\rightarrow +\infty$, by Theorem~\ref{thm:sindelta} (see Appendix)
\begin{equation}\label{eq:abOscil2}
\begin{split}
  {|b|}^2=&\frac{4{|V|}^2}{\hbar^2\Omega'^2}\sin^2\Omega't/2\\
         =&\frac{2{|V|}^2T}{\hbar^2}\pi\delta(\omega_{eg}-\omega)\\
         =&\frac{2\pi}{\hbar}{|V|}^2T\delta(E_e-E_g-\hbar\omega),
\end{split}
\end{equation}
which is exactly Fermi's golden rule.  In fact, in Ref.~\citenum{JCModel1963}, Jaynes and Cummings proposed their model to calculate the noise figures of maser beyond Fermi's golden rule.

Now we introduce the semiclassical version of the Jaynes-Cummings model that inherits most features of the full quantum version and is easier to apply to real systems of CQED and cold atoms.  The only difference is that now electromagnetic field is classical, i.e. $\mathbf{E}=\mathbf{E}_0 \cos (\omega t+\phi)$.  Then the interaction Hamiltonian in Eq.~\ref{eq:HamAtomFieldCoupling} is
\begin{equation}\label{eq:HamInt}
\begin{split}
  H_{int}=&-\hat{\mathbf{D}}\cdot\mathbf{E}=-\mathbf{d}_0 \cdot \mathbf{E}_0 \cos (\omega t+\phi)(\hat{\sigma}_+ + \hat{\sigma}_-) \\
  =&\hbar\Omega \cos (\omega t+\phi)(\hat{\sigma}_+ + \hat{\sigma}_-)\\
  =&\hbar\Omega \cos (\omega t+\phi)\hat{\sigma}_x,
\end{split}
\end{equation}
where $\Omega=-\mathbf{d}_0 \cdot \mathbf{E}_0 /\hbar$ is called Rabi frequency. As for the Hamiltonian of the two-level atom, by choosing a proper point of $E=0$, we may write it as
\begin{equation}\label{eq:HamAtom}
  H_{atom}=\frac{1}{2}\hbar\omega_{eg}\hat\sigma_z,
\end{equation}
where $\omega_{eg}$ is the energy difference of two internal levels.  The total Hamiltonian $H=H_{atom}+H_{int}$ is time dependent.  Since the system is not far from resonance, we can change to a ``rotating'' frame to eliminate the time dependence by choosing an ``interaction'' picture properly.
\begin{equation}\label{eq:HamSplit}
  \begin{split}
    H=& \frac{1}{2}\hbar\omega_{eg}\hat\sigma_z +\hbar\Omega \cos (\omega t+\phi)\hat{\sigma}_x \\
     =& \frac{1}{2}\hbar\omega \hat\sigma_z +\frac{1}{2}\hbar(\omega_{eg}-\omega)\hat\sigma_z + \hbar\Omega \cos (\omega t+\phi)\hat{\sigma}_x \\
     =&H_0+H_1,
  \end{split}
\end{equation}
where $H_0=\frac{1}{2}\hbar\omega \hat\sigma_z$ and $H_1=\frac{1}{2}\hbar(\omega_{eg}-\omega)\hat\sigma_z + \hbar\Omega \cos (\omega t+\phi)\hat{\sigma}_x$.  Then
\begin{equation}\label{eq:HamI}
\begin{split}
  H_I=&e^{iH_0t/\hbar}H_1e^{-iH_0t/\hbar}\\
     =&\frac{1}{2}\hbar\Delta \hat\sigma_z + \hbar\Omega \cos (\omega t+\phi)\\
      & \times e^{i\frac{1}{2}\omega t \sigma_z}(\hat{\sigma}_+ + \hat{\sigma}_-)e^{-i\frac{1}{2}\omega t \sigma_z}\\
     =&\frac{1}{2}\hbar\Delta \hat\sigma_z + \hbar\Omega \cos (\omega t+\phi)(e^{i\omega t}\hat{\sigma}_+ + e^{-i\omega t}\hat{\sigma}_-)\\
     =&\frac{1}{2}\hbar\Delta \hat\sigma_z + \frac{1}{2}\hbar\Omega \bigg[e^{-i\phi}\hat{\sigma}_+ + e^{i\phi}\hat{\sigma}_-\\
      &+ e^{i(2\omega t+\phi)}\hat{\sigma}_+ + e^{-i(2\omega t+\phi)}\hat{\sigma}_-\bigg]\\
     \approx&\frac{1}{2}\hbar\Delta \hat\sigma_z + \frac{1}{2}\hbar\Omega \bigg[e^{-i\phi}\hat{\sigma}_+ + e^{i\phi}\hat{\sigma}_-\bigg]\\
     =& \frac{\hbar}{2}
        \begin{pmatrix}
          \Delta                  &    e^{-i\phi} \Omega \\
          e^{i\phi} \Omega        &    -\Delta
        \end{pmatrix}\\
     =& \frac{\hbar \Omega'}{2}
        \begin{pmatrix}
          \cos \theta             &    e^{-i\phi} \sin \theta \\
          e^{i\phi} \sin \theta   &    -\cos \theta
        \end{pmatrix}
\end{split}
\end{equation}
where $\Delta=\omega_{eg}-\omega$, $\Omega'=\sqrt{\Delta^2+\Omega^2}$, $\theta=\arctan (\Omega/\Delta)$, and $\hbar\Omega e^{-i\phi}/2=V$.
\begin{figure}
  \centering
  \includegraphics[scale=0.45]{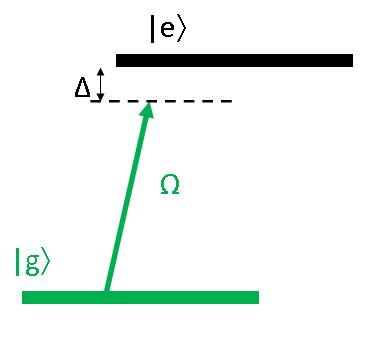}\\
  \caption{Two-Level Atom Coupling to a Light Field}\label{fig:TLAtom}
\end{figure}
The first term of $H_1$ always commutes with $H_0$.  We used $\hat\sigma_+\hat\sigma_z=(\hat\sigma_z-2)\hat\sigma_+$ and $\hat\sigma_-\hat\sigma_z=(\hat\sigma_z+2)\hat\sigma_-$, or Baker-Campbell-Hausdorff formula in the second step.  Omitting the fast oscillating terms in the fourth step is so called RWA.  The interacting Hamiltonian can also be written as $ H_I=\frac{1}{2}\hbar \Omega'\hat\sigma_{\mathbf{n}}$, where $\mathbf{n}=(\frac{\Omega}{\Omega'} \cos\phi, \frac{\Omega}{\Omega'} \sin\phi, \frac{\Delta}{\Omega'})$, which has an easy geometrical interpretation.  In the interaction picture, the Bloch vector corresponding to the state rotates around axis in direction of $\mathbf{n}$ by angular velocity $\Omega'$.  For $\Delta \ll \omega\approx\omega_{eg}$, the result of semiclassical approach agrees with the one obtained by full quantum method.

So far, we neglected the orbital motion of the atom in the real space.  The phase of the laster $\phi$ and detuning angle $\theta$ are space dependent.  When the atom moves in the light field, they will influence its orbital motions.  We shall show bellow that this kind of influence can be described as a vector potential in adiabatic limit. For two level system, with initial condition in one of the internal eigenstate, this vector potential is an Abelian gauge potential; the atom motion was modified as if there is a magnetic field.

\section{Artificial Gauge fields in Cold Atoms}\label{sec:ArtificialGauge}
\subsection{Tow Level Atom in a Light Beam}\label{sec:TLA}
As we have shown in the previous sections, two-level atoms can be described by Jaynes-Cummings's model.  We may denote \{$|g\rangle$, $|e\rangle$\}  a basis of the two-dimensional Hilbert space of the internal degree of freedom. The coupling of the internal degree of freedom with the light field under RWA can be described by a 2 by 2 matrix of the form
\begin{equation}\label{eq:JCRabi}
U=\frac{\hbar \Omega'}{2}
\begin{pmatrix}
  \cos \theta             &    e^{-i\phi} \sin \theta \\
  e^{i\phi} \sin \theta   &    -\cos \theta
\end{pmatrix},
\end{equation}
where $\Omega'$ is the generalized Rabi frequency, $\theta$ and $\phi$ are two position dependent angle parameters (as shown in Eq.~\ref{eq:HamI}). Therefore the total Hamiltonian of the atom moving in the light beam is
\begin{equation}\label{eq:totHam}
  H=\left( \frac{\mathbf{P}^2}{2M} +V \right)\hat{\mathbb{I}} +U,
\end{equation}
where $M$ is the mass of the atom, $\mathbf{P}=-i\hbar\nabla$ is the total momentum operator and $\hat{\mathbb{I}}$ is the identity in internal space.\\

At any point $\mathbf{r}$, $U$ has two eigenstates
\begin{equation}\label{eq:inteignvect}
  \begin{split}
  |\chi_1\rangle =&
    \begin{pmatrix}
      \cos (\theta/2)\\
      e^{i\phi} \sin (\theta/2)
    \end{pmatrix}, \\
  |\chi_2\rangle =&
    \begin{pmatrix}
      -e^{-i\phi} \sin (\theta/2)\\
      \cos (\theta/2)
    \end{pmatrix},
  \end{split}
\end{equation}
with eigenvalues $\hbar\Omega'/2$ and $-\hbar\Omega'/2$ respectively
\footnote{For brevity of notions we may denote simply $\Omega$ for $\Omega'$ here as long as there is no ambiguity of the meaning of Rabi frequency; and in most cases we consider in this article are of near resonance $\Delta\ll\omega\sim\omega_{eg}$. }.  A remark concerning state bases of $\{|e,\tilde{\mathbf{E}}\rangle, |g,\tilde{\mathbf{E}'}\rangle\}$ and $\{|\chi_1\rangle, |\chi_2\rangle\}$ deserves to be emphasized here.  We have shown in the previous section, when the frequency of the laser is near resonance and dipole coupling is strong, the off diagonal elements in Eq.~\ref{eq:JCRabi} drive the internal states of the atom oscillating back and forth with absorbing and emitting a photon in resonance.  Therefore, $|e\rangle$ and $|g\rangle$ are no longer convenient basis to describe the internal state of the atoms.  On the contrary, states given in Eq.~\ref{eq:inteignvect} are eigenvectors of $U$, thus we can diagonalize $U$ by a unitary transform $S$ from the basis of $\{|e,\tilde{\mathbf{E}}\rangle, |g,\tilde{\mathbf{E}'}\rangle\}$ to $\{|\chi_1\rangle, |\chi_2\rangle\}$, and then the total Hamiltonian is diagonal in internal space.  This perspective is also called ``Dressed-Atom Approach''~\cite{Dalibard1985}.  When the coupling decrease to zero, the dressed states come back to $|e\rangle$ and $|g\rangle$.

However, from Eq.~\ref{eq:inteignvect}, we may observe these dressed states are position dependent. When the atom is coupling with the field, they are not degenerate (separated by a finite gap $2|V|$). Under adiabatic limit, if the atom is in one of the dressed eigenstate and the orbital motion of atom is slow enough, it shall always remain in the submanifold of that state, while position-dependent projections will contribute to the orbital wave function a phase factor in addition to the ordinary dynamical phase factor, which leads to the gauge field at last.

\begin{widetext}
Now let us expand the total wave function in term of these local basis as
\begin{equation}\label{eq:totwavefunc}
  |\Psi (\mathbf{r},t)\rangle = \sum_{j=1,2}\psi_j (\mathbf{r},t)|\chi_j\rangle,
\end{equation}
where $\psi_j (\mathbf{r},t),\ j=1,2$ are time dependent wave functions
\footnote{In case the clumsy summing signs, we may adopt Einstein's convention below; while the summing signs are sometimes written explicitly for carefulness. Both sub- and superscripts are used for convenience; since we do live well in three dimentional euclidean space with standard metric, they are in fact the same thing.}
.
Apply the total Hamiltonian (Eq.~\ref{eq:totHam}) to our total wave function, we have our time dependent Schr\"odinger equation
\begin{equation}\label{eq:totSchroeEq}
\begin{split}
  i\hbar \frac{\partial}{\partial t} |\Psi (\mathbf{r},t)\rangle
   =&H |\Psi (\mathbf{r},t)\rangle = \left[ \left( \frac{\mathbf{P}^2}{2M} +V \right)\hat{\mathbb{I}} +U \right]|\Psi (\mathbf{r},t)\rangle.
\end{split}
\end{equation}
Since the eigenstates $|\chi_j\rangle$ of $U$ are position dependent, $\mathbf{P}$ may act on both part of the wave function, differentiating with respect to $\mathbf{r}$.  The second term $V$ represents trap potential that is diagonal both in internal space and real space (we shall assume that $|e\rangle$ and $|g\rangle$ experience same trap potential).  The third term acts on the internal eigenstates and gives the corresponding eigenenergies.\\

In the first step, we shall calculate the total momentum operator $\mathbf{P}$ acting on the total wave function $|\Psi (\mathbf{r},t)\rangle$ .

\begin{equation}\label{eq:momentumdiff}
  \begin{split}
    \mathbf{P}|\Psi (\mathbf{r},t)\rangle
    =& -i\hbar\nabla
      \left(\sum_{j=1,2}\psi_j(\mathbf{r},t)|\chi_j\rangle \right)
    = -i\hbar \sum_{j=1,2}
      [\nabla\psi_j(\mathbf{r},t)]|\chi_j\rangle
      +\psi_j(\mathbf{r},t)(\nabla|\chi_j\rangle)\\
    =& -i\hbar \sum_{j,l=1,2}
      [\nabla\psi_j(\mathbf{r},t)]|\chi_l\rangle
      \langle \chi_l |\chi_j\rangle 
      +\psi_j(\mathbf{r},t)|\chi_l\rangle \langle \chi_l| \nabla|\chi_j\rangle\\
    =&\sum_{j,l=1,2}
      [(\mathbf{p}\delta_{lj}-\mathbf{A}_{lj})\psi_j(\mathbf{r},t)]
      |\chi_l\rangle\\
  \end{split}
\end{equation}

where $\mathbf{p}$ is of the form of momentum operator but not acting on the spinor, and $\mathbf{A}_{lj}=i\hbar \langle \chi_l | \nabla|\chi_j\rangle$ is the gauge potential, or in a sense of mathematics, the ``connection''.

Then we can calculate the kinetic energy.

\begin{equation}\label{eq:kineticEtot}
  \begin{split}
    \frac{\mathbf{P}^2}{2M}|\Psi (\mathbf{r},t)\rangle
    =&\frac{1}{2M}\mathbf{P} \left( \sum_{j,l=1,2}
      [(\mathbf{p}\delta_{lj}-\mathbf{A}_{lj})\psi_j(\mathbf{r},t)]
      |\chi_l\rangle \right) \\
    =&\frac{1}{2M}\sum_{j,l,m=1,2} \Bigg(
      \{\mathbf{p}\delta_{ml}[(\mathbf{p}\delta_{lj}-\mathbf{A}_{lj})
      \psi_j(\mathbf{r},t)]\}
      |\chi_m\rangle
      - \{\mathbf{A}_{ml}[(\mathbf{p}\delta_{lj}-\mathbf{A}_{lj})
     \psi_j(\mathbf{r},t)]\}
      |\chi_m\rangle \Bigg)\\
    =&\frac{1}{2M}\sum_{j,l,m=1,2} \Bigg(
      \{(\mathbf{p}\delta_{ml} - \mathbf{A}_{ml} )
      [(\mathbf{p}\delta_{lj}-\mathbf{A}_{lj})\psi_j(\mathbf{r},t)]\}
      |\chi_m\rangle \Bigg),\\
  \end{split}
\end{equation}
\end{widetext}
where $\mathbf{p}$ still only acts on the ``orbital'' part.
Now, projecting it to the eigen space of $\chi_n\rangle$, and rewrite the formula in the form of matrix,
\begin{widetext}
\begin{equation}\label{eq:projchin}
  \begin{split}
    \langle \chi_n| \frac{\mathbf{P}^2}{2M}|\Psi (\mathbf{r},t)\rangle
    =& \frac{1}{2M}\sum_{j,l=1,2} \Bigg(
      \{(\mathbf{p}\delta_{nl} - \mathbf{A}_{nl} )
      (\mathbf{p}\delta_{lj}-\mathbf{A}_{lj})\psi_j(\mathbf{r},t)\}
      \Bigg), \ \text{or}\\
    \langle \chi_n| \frac{\mathbf{P}^2}{2M}|\Psi (\mathbf{r},t)\rangle
    =&  {\Bigg(\frac{{(\mathbf{p}\hat{\mathbb{I}}- \hat{\mathbf{A}})}^2}{2M}\Phi(\mathbf{r},t)
      \Bigg)}_{n},
  \end{split}
\end{equation}
\end{widetext}
where $\hat{\mathbf{A}}$ is a matrix vector, $\Phi=
\begin{pmatrix}
\psi_1(\mathbf{r},t)\\
\psi_2(\mathbf{r},t)
\end{pmatrix}$ is a two component wave function of orbital part, and the subscript $n$ at down right corner indicates the n-th component of the wave function.

The second term and the third term of Eqn.~\ref{eq:totSchroeEq} is easy to obtain
and project to $|\chi_n\rangle$
\begin{equation}\label{eq:projchinpotetspin}
  \begin{split}
    \langle \chi_n| V\hat{\mathbb{I}} |\Psi (\mathbf{r},t)\rangle
    = & {\Bigg( V(\mathbf{r})\Phi(\mathbf{r},t) \Bigg)}_n,\\
     \langle \chi_n| U |\Psi (\mathbf{r},t)\rangle
    = & {\Bigg( \frac{\hbar\Omega}{2}\hat\sigma_z \Phi(\mathbf{r},t) \Bigg)}_n,\\
    &\phantom{=}
  \end{split}
\end{equation}
where $n=1,2$.
Now let us calculate matrix vector $\hat{\mathbf{A}}$.
According to our definition $\mathbf{A}_{lj}=i\hbar \langle \chi_l | \nabla|\chi_j\rangle$, the a-th component of $\hat{\mathbf{A}}$ is
\begin{equation*}
  {\hat{\mathbf{A}}}_a=i\hbar \langle \chi_l | \partial_a |\chi_j\rangle=i\hbar
  \begin{pmatrix}
  \langle \chi_1 | \partial_a |\chi_1\rangle & \langle \chi_1 | \partial_a |\chi_2\rangle \\
  \langle \chi_2 | \partial_a |\chi_1\rangle & \langle \chi_2 | \partial_a |\chi_2\rangle
  \end{pmatrix}.
\end{equation*}
We shall always use $a,b,c...$ for the indices of vector components, and $i,j,k,l,m,n...$ for indices of matrix.
Its entries are
\begin{widetext}
\begin{equation}\label{eq:entriesA}
  \begin{split}
    i\hbar \langle \chi_1 | \partial_a |\chi_1\rangle =& i\hbar
      \begin{pmatrix} \cos \frac{\theta}{2} &  e^{-i\phi} \sin \frac{\theta}{2}\end{pmatrix}
      \left(
      \begin{pmatrix}
        -\sin \frac{\theta}{2} \\
        e^{i\phi} \cos \frac{\theta}{2}
      \end{pmatrix}\left[\frac{\partial_a \theta}{2}\right]
      +\begin{pmatrix}
        0\\
        e^{i\phi} \sin \frac{\theta}{2}
      \end{pmatrix}[{i\partial_a \phi}]\right)\\
      =& -\hbar \sin^2 \left(\frac{\theta}{2}\right) [\partial_a \phi]= \frac{\hbar}{2} (\cos {\theta}-1) [\partial_a \phi]\\
    i\hbar \langle \chi_2 | \partial_a |\chi_2\rangle =& i\hbar
      \begin{pmatrix} -e^{i\phi} \sin \frac{\theta}{2} & \cos \frac{\theta}{2} \end{pmatrix}
      \left(
      \begin{pmatrix}
        -e^{-i\phi} \cos \frac{\theta}{2}\\
        -\sin \frac{\theta}{2}
      \end{pmatrix}\left[\frac{\partial_a \theta}{2}\right]
      +\begin{pmatrix}
        e^{-i\phi} \sin \frac{\theta}{2}\\
        0
      \end{pmatrix}[{i\partial_a \phi}]\right)\\
      =& \hbar \sin^2 \left(\frac{\theta}{2}\right) [\partial_a \phi]= \frac{\hbar}{2} (1-\cos {\theta}) [\partial_a \phi]\\
    i\hbar \langle \chi_1 | \partial_a |\chi_2\rangle =& i\hbar
      \begin{pmatrix} \cos \frac{\theta}{2} &  e^{-i\phi} \sin \frac{\theta}{2}\end{pmatrix}
      \left(
      \begin{pmatrix}
        -e^{-i\phi} \cos \frac{\theta}{2}\\
        -\sin \frac{\theta}{2}
      \end{pmatrix}\left[\frac{\partial_a \theta}{2}\right]
      +\begin{pmatrix}
        e^{-i\phi} \sin \frac{\theta}{2}\\
        0
      \end{pmatrix}[{i\partial_a \phi}]\right)\\
      =& -i\hbar e^{-i\phi}\left[\frac{\partial_a \theta}{2}\right]+i\hbar e^{-i\phi}\cos \frac{\theta}{2}\sin \frac{\theta}{2}[i\partial_a \phi]\\
      =&-\frac{\hbar}{2} e^{-i\phi}
      (i\partial_a \theta +(\sin\theta )\partial_a \phi)\\
    i\hbar \langle \chi_2 | \partial_a |\chi_1\rangle =&{\Big( -i\hbar \langle  \partial_a \chi_1 | \chi_2\rangle \Big)}^*
    ={\Big( i\hbar \langle  \chi_1 | \partial_a |\chi_2\rangle \Big)}^*.
  \end{split}
\end{equation}
So
\begin{equation}\label{eq:matA}
\begin{split}
{\hat{\mathbf{A}}}_a=&
\begin{pmatrix}
  -\hbar \sin^2 \left(\frac{\theta}{2}\right) [\partial_a \phi] &
  -\frac{\hbar}{2} e^{-i\phi} (i\partial_a \theta +(\sin\theta )\partial_a \phi)\\
  \frac{\hbar}{2} e^{i\phi}(i\partial_a \theta -(\sin\theta )\partial_a \phi)
  & \hbar \sin^2 \left(\frac{\theta}{2}\right) [\partial_a \phi]
\end{pmatrix}\\
  =&-\frac{\hbar}{2} ((\sin\phi)\partial_a \theta +(\sin\theta \cos \phi)\partial_a \phi) \hat{\sigma}_x
  +\frac{\hbar}{2} ((\cos \phi)\partial_a \theta  -(\sin\theta \sin \phi)\partial_a \phi) \hat{\sigma}_y
  -\hbar \sin^2 \left(\frac{\theta}{2}\right) [\partial_a \phi] \hat{\sigma}_z
\end{split}
\end{equation}
who is Hermitian.

Furthermore
\begin{equation*}
  \begin{split}
  \ {\hat{\mathbf{A}}}_a^2=&\begin{pmatrix}
    -\hbar \sin^2 \left(\frac{\theta}{2}\right) [\partial_a \phi]
    & -\frac{\hbar}{2} e^{-i\phi} (i\partial_a \theta +(\sin\theta )\partial_a \phi)\\
    \frac{\hbar}{2} e^{i\phi}(i\partial_a \theta -(\sin\theta )\partial_a \phi)
    & \hbar \sin^2 \left(\frac{\theta}{2}\right) [\partial_a \phi]
  \end{pmatrix}^2\\
  =&\begin{pmatrix}
    \hbar^2 \sin^4 \left(\frac{\theta}{2}\right){[\partial_a \phi]}^2+\frac{\hbar^2}{4}({(\partial_a \theta)}^2+\sin^2 \theta {[\partial_a \phi]}^2)
    & 0\\
    0
    & \hbar^2 \sin^4 \left(\frac{\theta}{2}\right){[\partial_a \phi]}^2+\frac{\hbar^2}{4}({(\partial_a \theta)}^2+\sin^2 \theta {[\partial_a \phi]}^2)
  \end{pmatrix}\\
  =&\begin{pmatrix}
    \hbar^2 \sin^2 \left(\frac{\theta}{2}\right){[\partial_a \phi]}^2+\frac{\hbar^2}{4}{(\partial_a \theta)}^2
    & 0\\
    0
    & \hbar^2 \sin^2 \left(\frac{\theta}{2}\right){[\partial_a \phi]}^2+\frac{\hbar^2}{4}{(\partial_a \theta)}^2
  \end{pmatrix}
  \end{split}
\end{equation*}
and
\begin{equation}\label{eq:matAsq}
\begin{split}
  {{\hat{\mathbf{A}}}}^2
  =&\sum_{a=1}^{3}{{\hat{\mathbf{A}}}_a}^2
  =\left( \hbar^2 \sin^2\left(\frac{\theta}{2}\right){|\nabla \phi|}^2+\frac{\hbar^2}{4}{|\nabla \theta|}^2\right)\hat{\mathbb{I}}\\
  {{\hat{\mathbf{A}}}}^2_{\phantom{2}jk}
  =&{\hat{\mathbf{A}}}_{jl}\cdot{\hat{\mathbf{A}}}_{lk}
\end{split}
\end{equation}


Thus, with Eq.~\ref{eq:totwavefunc}, and projecting the left hand side of Eq.~\ref{eq:totSchroeEq} to $|\chi_n\rangle$, we have
\begin{equation}\label{eq:resultl}
  \begin{split}
     \langle \chi_n| i\hbar \frac{\partial}{\partial t} |\Psi (\mathbf{r},t)\rangle
     =&\langle \chi_n| i\hbar \frac{\partial}{\partial t}  \left(\sum_{j=1,2}\psi_j (\mathbf{r},t)|\chi_j\rangle \right)
     =\langle \chi_n| \left(\sum_{j=1,2}[i\hbar \frac{\partial}{\partial t} \psi_j (\mathbf{r},t)]|\chi_j\rangle \right)\\
     =&i\hbar \frac{\partial}{\partial t} \psi_n (\mathbf{r},t)
     =i\hbar \frac{\partial}{\partial t}
       {\Bigg(\Phi(\mathbf{r},t)\Bigg)}_n,
  \end{split}
\end{equation}
where adiabatic assumption has been used\footnote{So far, every thing is exact. The adiabatic assumption just suggests that under this projection, if initial state is one of the energy eigenstates, it will always remain in this submanifold.}.
On the other hand, with the help of Eqn.~\ref{eq:projchin} and ~\ref{eq:projchinpotetspin}, we have
\begin{equation}\label{eq:resultr}
  \begin{split}
     \langle \chi_n| H |\Psi (\mathbf{r},t)\rangle
     =&\langle \chi_n|  \left[ \left( \frac{\mathbf{P}^2}{2M} +V \right)\hat{\mathbb{I}} +U \right]|\Psi (\mathbf{r},t)\rangle\\
     =&\langle \chi_n| \frac{\mathbf{P}^2}{2M}|\Psi (\mathbf{r},t)\rangle
       +\langle \chi_n| V 
         |\Psi (\mathbf{r},t)\rangle
       + \langle \chi_n| U |\Psi (\mathbf{r},t)\rangle\\
     =&{\Bigg(\frac{{(\mathbf{p}\hat{\mathbb{I}}
         -\hat{\mathbf{A}})}^2}{2M}\Phi(\mathbf{r},t)\Bigg)}_{n}
      +{\Bigg( V(\mathbf{r})\Phi(\mathbf{r},t) \Bigg)}_n
      +{\Bigg( \frac{\hbar\Omega}{2}\sigma_z
         \Phi(\mathbf{r},t)\Bigg)}_n.
  \end{split}
\end{equation}
Combining Eq.~\ref{eq:resultl} and ~\ref{eq:resultr}, we arrive at
\begin{equation}\label{eq:result1}
  \begin{split}
     i\hbar \frac{\partial}{\partial t}
       \begin{pmatrix}\psi_1(\mathbf{r},t)\\ \psi_2(\mathbf{r},t) \end{pmatrix}
     =&\frac{{(\mathbf{p}\hat{\mathbb{I}}
         -\hat{\mathbf{A}})}^2}{2M}\begin{pmatrix}\psi_1(\mathbf{r},t)\\ \psi_2(\mathbf{r},t) \end{pmatrix}
      + V(\mathbf{r})\begin{pmatrix}\psi_1(\mathbf{r},t)\\ \psi_2(\mathbf{r},t) \end{pmatrix}
      + \frac{\hbar\Omega}{2}\sigma_z
         \begin{pmatrix}\psi_1(\mathbf{r},t)\\ \psi_2(\mathbf{r},t) \end{pmatrix}\\
     =&\frac{{(\mathbf{p}\hat{\mathbb{I}}
         -\hat{\mathbf{A}})}^2}{2M}\begin{pmatrix}\psi_1(\mathbf{r},t)\\ \psi_2(\mathbf{r},t) \end{pmatrix}
      + \begin{pmatrix}V(\mathbf{r})& 0\\0 & V(\mathbf{r}) \end{pmatrix}
      \begin{pmatrix}\psi_1(\mathbf{r},t)\\ \psi_2(\mathbf{r},t) \end{pmatrix}
      + \begin{pmatrix}\frac{\hbar\Omega}{2}& 0\\0 & -\frac{\hbar\Omega}{2} \end{pmatrix}
         \begin{pmatrix}\psi_1(\mathbf{r},t)\\ \psi_2(\mathbf{r},t) \end{pmatrix}.
  \end{split}
\end{equation}
It is obvious to observe that the last two terms on the right hand side of Eqn.~\ref{eq:result1} are diagonal, while carefulness is needed to calculate the first term, since $\mathbf{p}$ does not commute with $\hat{\mathbf{A}}$.  Without lose of generality, we may observed the first component of the wave function.  Although $\hat{\mathbf{A}}^2$ is diagonal, $\hat{\mathbf{A}}$ itself is not, thus two components may be well coupled through the off-diagonal entries of $\hat{\mathbf{A}}$. \\

We shall now calculate the first term explicitly
\footnote{One should pay attention to the calculations. $\mathbf{p}$ is not only diagonal matrix, but also a differential operator who acts both on wave function and $\hat{\mathbf{A}}$.  We shall well keep their order during the calculation.}.\\

Notice that
\begin{equation}\label{eq:pAApphi}
  \begin{split}
    (\mathbf{p}\cdot\hat{\mathbf{A}}+\hat{\mathbf{A}}\cdot\mathbf{p})\Phi
    =&\left(
    \begin{pmatrix}
      \mathbf{p} \cdot \mathbf{A}_{11}
      &\mathbf{p}\cdot \mathbf{A}_{12}\\
      \mathbf{p} \cdot \mathbf{A}_{21}
      &\mathbf{p}\cdot \mathbf{A}_{22}
    \end{pmatrix}+
    \begin{pmatrix}
       \mathbf{A}_{11} \cdot \mathbf{p}
       &\mathbf{A}_{12}\cdot \mathbf{p}\\
       \mathbf{A}_{21} \cdot \mathbf{p}
       &\mathbf{A}_{22}\cdot \mathbf{p}
    \end{pmatrix}
      \right)
      \begin{pmatrix}\psi_1(\mathbf{r},t)\\ \psi_2(\mathbf{r},t)\end{pmatrix}\\
    =&\begin{pmatrix}
    (\mathbf{p} \cdot \mathbf{A}_{11}+\mathbf{A}_{11} \cdot \mathbf{p}) \psi_1(\mathbf{r},t) +
    (\mathbf{p} \cdot \mathbf{A}_{12}+\mathbf{A}_{12} \cdot \mathbf{p}) \psi_2(\mathbf{r},t) \\
    (\mathbf{p} \cdot \mathbf{A}_{21}+\mathbf{A}_{21} \cdot \mathbf{p}) \psi_1(\mathbf{r},t) +
    (\mathbf{p} \cdot \mathbf{A}_{22}+\mathbf{A}_{22} \cdot \mathbf{p}) \psi_2(\mathbf{r},t) \\\end{pmatrix}\\
  \end{split}
\end{equation}

\begin{equation}\label{eq:redKE}
\begin{split}
  \frac{{(\mathbf{p}\hat{\mathbb{I}}-\hat{\mathbf{A}})}^2}{2M}
     \begin{pmatrix}\psi_1(\mathbf{r},t)\\ \psi_2(\mathbf{r},t)\end{pmatrix}
   =&\frac{1}{2M}{(\mathbf{p}^2
   -\mathbf{p}\cdot\hat{\mathbf{A}}
   -\hat{\mathbf{A}}\cdot\mathbf{p}
   +\hat{\mathbf{A}}^2)} \begin{pmatrix}\psi_1(\mathbf{r},t)\\ \psi_2(\mathbf{r},t)\end{pmatrix}\\
   =&\frac{1}{2M}
   \begin{pmatrix}
   {(\mathbf{p}-\mathbf{A}_{11})}^2-\mathbf{A}_{11}\cdot\mathbf{A}_{11}
   +{(\hat{\mathbf{A}}^2)}_{\phantom{a}11}
   & -(\mathbf{p}\cdot\mathbf{A}_{12}+\mathbf{A}_{12}\cdot\mathbf{p})\\
   -(\mathbf{p}\cdot\mathbf{A}_{21}+\mathbf{A}_{21}\cdot\mathbf{p})
   &{(\mathbf{p}-\mathbf{A}_{22})}^2-\mathbf{A}_{22}\cdot\mathbf{A}_{22}
   +{(\hat{\mathbf{A}}^2)}_{\phantom{a}22}
   \end{pmatrix}
   \begin{pmatrix}\psi_1(\mathbf{r},t)\\ \psi_2(\mathbf{r},t)\end{pmatrix},
\end{split}
\end{equation}
\end{widetext}

With Eq.~\ref{eq:entriesA} or Eq.~\ref{eq:matAsq}, we have diagonal term of Eq.~\ref{eq:redKE} are
\[
\frac{{(\mathbf{p}-\mathbf{A}_{ii})}^2}{2M}+W,
\]
where
$W=\frac{1}{2M}{({(\hat{\mathbf{A}}^2)}_{\phantom{a}ii}-\mathbf{A}_{ii}\cdot\mathbf{A}_{ii})}
=\frac{\hbar^2}{8M}({|\nabla\theta|}^2+{\sin^2\theta|\nabla\phi|}^2)$.
The off-diagonal terms
\begin{equation}\label{eq:couple}
\begin{split}
  -\frac{1}{2M}(\mathbf{p}\cdot\mathbf{A}_{12}
  +\mathbf{A}_{12}\cdot\mathbf{p})
\end{split}
\end{equation}
 couple the two components of the wave function $(\Phi)$.

If we further assume that initially $\psi_2=0$, and take the adiabatic limit, we can arrive at Schr\"odinger equation for the wave function in $|\chi_1\rangle$ subspace
\begin{equation}\label{eq:redResult}
  i\hbar \frac{\partial \psi_1}{\partial t}=
  \left[
  \frac{{(\mathbf{p}-\mathbf{A})}^2}{2M}+V+\frac{\hbar \Omega}{2}+W
  \right] \psi_1,
\end{equation}
where $\mathbf{A}=\mathbf{A}_{11}$ and $\mathbf{A}$ plays a role of artificial Abelian gauge field.

A remark concerning the adiabatic approximation deserves here. By definition, off-diagonal elements of connection $\mathbf{A}_{ij}=i\hbar \langle \chi_i |\nabla | \chi_j \rangle \sim i \hbar {\langle \chi_i |\frac{\partial H}{\partial \mathbf{r}} | \chi_j \rangle}/{(E_j-E_i)}$
\footnote{$\langle \chi_i |\frac{\partial H}{\partial \mathbf{r}}| \chi_j \rangle = {\langle \chi_i | \nabla | \chi_j \rangle}{(E_j-E_i)}$, see Eq.~8 of Ref.~\citenum{Berry1984}}.  Thus the coupling term in Eq.~\ref{eq:couple} is of order of $\mathbf{v}\cdot\mathbf{A}_{ij}\sim i\hbar {\langle \chi_i |\frac{\partial H}{\partial t} | \chi_j \rangle}/{(E_j-E_i)}$ that vanishes as first order in adiabatic limit. Since the initial value of $\psi_2$ is zero, its amplitude also vanishes as first order, and its contribution to $\psi_1$ through the off-diagonal terms vanishes as second order under adiabatic limit, which validates our adiabatic approximation.  Therefore projecting to one of the submanifolds, we arrive at the non-Abelian gauge theory as shown in Eq.~\ref{eq:redResult}.  This estimation also gives a criteria for adiabatic approximation in real experiments.  The typical orbital energy changes of order $\delta V=V(\mathbf{r})-V(\mathbf{r}+\mathbf{v}\delta t)$ that should be much smaller than the gap $E_2-E_1$ of the dressed-states.  This means that our atom should move ``slow'' enough, which is only a ``na\"ive'' sense of adiabatic approximation.

A further remark is that, in fact, our toy two-level dressed-atom model did exhibit full non-Abelian characters that is obvious from Eq.~\ref{eq:result1} (an example of the transformation of the connection $\mathbf{A}$ is verified in Proposition~\ref{thm:nonabtrans} in appendix).  However, our adiabatic limit guarantees that if initially our system is in an eigenstate of dressed state, it shall always remain in that sub-manifold.  But if the gap of two states are not large enough, this simple projection will no longer be valid, while the non-Abelian gauge features in Eq.~\ref{eq:result1} still hold.  More generally, if our system is initially in a sub-manifold spanned by several quasi degenerate states whereas this sub-manifold is well separated from other sub-manifolds by large gaps, which enables us to apply adiabatic projection, then we can construct artificially  non-Abelian schemes in our atom-light systems.

Unfortunately, the scheme of our toy model is not very practical.  The excited state of the atom may spontaneously decay.  The collisions between the atom will also cause some problems.  In the following sections, we shall introduce several more practical schemes for realizing artificial gauge fields in cold atoms, and we shall also analyze carefully the problems we may encountered.  Nevertheless, the idea is just a simple generalization that shares almost all the key factors that we have shown in our toy model.

\subsection{One, Two, Three to N...}\label{sec:NonAbelian}
\begin{figure}
  \centering
  \includegraphics[scale=0.32]{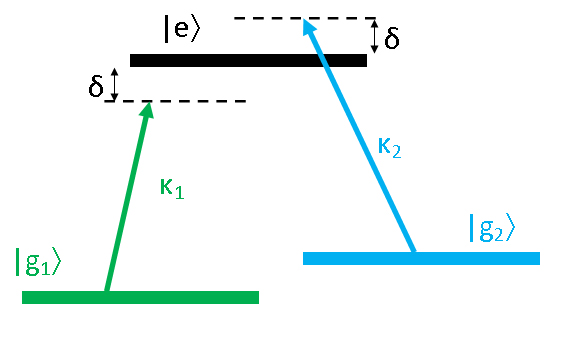}\\
  \caption{Atomic $\Lambda$-level Structure}\label{fig:lambdascheme}
\end{figure}
In previous section, we coupled one ground state to an excited state by dipole interaction, where the bare states of atom are dressed by the light field.  When we project to the sub-manifold of dressed states with adiabatic approximation, we obtain an artificial Abelian gauge field induced by the  space-dependent Rabi frequency and detuning.  We can simply generalize our scheme by coupling two quasi-degenerate ground states to another state.  The scheme was shown in Fig.~\ref{fig:lambdascheme}, where two state $|g_1\rangle$ and $|g_2\rangle$ are coupled to $|e\rangle$ by two laser beams. The laser couplings are subjected to selection rules, so we can choose proper polarization of lasers to control them separately.  As shown in previous section, the total coupling Hamiltonian under RWA can be written as
\begin{equation}\label{eq:Ham3level}
  H_I=\frac{\hbar}{2}
  \begin{pmatrix}
   -2\delta        &    \kappa_1^*       &   \ 0\\
   \kappa_1        &    0                &   \ \kappa_2\\
   0               &    \kappa_2^*       &   \phantom{-}2\delta
  \end{pmatrix},
\end{equation}
where $\kappa_i, i=1,2$ are complex space dependent Rabi frequencies, $\pm2\delta$ are detunings of photon excitation with respect to the Raman resonance~\cite{DalibardArtificialGauge}.
Now we consider the resonant case, i.e. $2\delta=0$.  Rewrite this Hamiltonian in Dirac bra-ket notation
\begin{equation*}
  \begin{split}
    H_I=&\frac{\hbar}{2}(\kappa_1|e\rangle\langle g_1|
        +\kappa_2|e\rangle\langle g_2|)+h.c. \\
       =& \frac{\hbar}{2}(\kappa|e\rangle\langle B|)+h.c. ,
  \end{split}
\end{equation*}
where $|B\rangle=(\kappa_1^*|g_1\rangle+\kappa_2^*|g_2\rangle)/\kappa$, and $\kappa=\sqrt{{|\kappa_1|}^2+{|\kappa_2|}^2}$.  By this notation, it is obvious that the interaction only couples $|e\rangle$ and $|B\rangle$ together. But now we have a Hilbert space spanned by three independent state vector.  So the third state
\begin{equation}\label{eq:Darkstate}
  |D\rangle=(\kappa_2|g_1\rangle-\kappa_1|g_2\rangle)/\kappa,
\end{equation}
who is orthogonal and uncoupled to $|e\rangle$ and $|B\rangle$ subspace has an eigenenergy of zero, which is called \emph{Dark State}.  The other two eigenstates span the same subspace of $<|e\rangle, |B\rangle>$.  In this submanifold, the coupling Hamiltonian reduces to a $2\times2$ matrix that is completely same as the one we solved in the previous section.  By diagonalize the Hamiltonian, we find they are $|\pm\rangle=(|e\rangle+|B\rangle)/\sqrt{2}$ with eigenenergies $E_\pm=\pm\hbar\kappa/2$ respectively.  And $|B\rangle$ is called \emph{Bright State}.

Since $|D\rangle$ is well separated from the other two eigenstates by an energy gap of $E_g=\hbar\kappa/2$, by adiabatic approximation and projection, it plays the same role as the dressed states in previous section. On the other hand, $|D\rangle$ is orthogonal to $|e\rangle$, i.e. no population in is $|e\rangle$, thus atom state is not affected by spontaneous emission.  In fact, this remarkable property was already well-known in quantum optics such as subrecoil cooling, Electromagnetically Induced Transparency (EIT) and STImulated Raman Adiabatic Passage(STIRAP).  These applications rely on the robustness of $|D\rangle$ with respect to the decoherence caused by spontaneous emission~\cite{DalibardArtificialGauge}.

By little laborious calculation as we have done in previous section, we obtained effective Equation of Motion for orbital part projected to $|D\rangle$ sub-manifold
\begin{equation}\label{eq:EoMLambda}
  i\hbar \frac{\partial \psi_D}{\partial t}
  =\left[
  \frac{{(\mathbf{p}-\mathbf{A})}^2}{2M}+V+W
  \right] \psi_D,
\end{equation}
where $\mathbf{A}=i\hbar\langle D|\nabla|D\rangle$  and $W=\hbar^2{|\langle B| \nabla |D\rangle |}^2/2M$ are the effective vector and scalar potential induced by the space-dependent dark state.

For pedagogical purposes, we do the calculations in detail again, but by the language of differential form for comparison.  The antisymmetrical properties and compact form of exterior differential operator will save us from the tedious component notations and simplify our calculations.

First, project our total wave function to the eigen basis of dressed atom as in Eq.~\ref{eq:totHam}
\begin{equation}\label{eq:3Dproj}
  |\Psi (\mathbf{r},t)\rangle = \sum_{j=D,\pm}\psi_j (\mathbf{r},t)|\chi_j\rangle.
\end{equation}
Since $P$ acts on both part of the wave function, as shown in Eq.~\ref{eq:momentumdiff}, we can easily arrive at the EoM of the form in analog to Eq.~\ref{eq:result1}, but the wave function has three components and all the matrices are $3\times3$.  In analog to Eq.~\ref{eq:redKE}, the off-diagonal terms coupling $|D\rangle$ to the other states are neglected due to the adiabatic approximation. Diagonal terms of $\hat{\mathbf{A}}$ and $\hat{\mathbf{A}}^2$ contribute to the vector and scalar potential.

Mathmatically, adiabatic assumption enables us to project the total wave function to the local basis of eigenstates of atom-light couplings.  This in a sense define a frame bundle $F$ over the adiabatic-parameter manifold (here the adiabatic parameter, also called base space or base manifold in mathematics, is the real position, dimension $d=2 \ or\  3$ for the cases discussed in this article).  Local frame is $n$-dimensional ($n=2$ for our toy model, $n=3$ for $\Lambda$-scheme, and $n=N+1$ for the N-pod scheme to be discussed). $\mathbf{P}/(-i\hbar)$ is the derivative operator.  As it acts on the orbital part, it becomes $\mathbf{p}/(-i\hbar)$ as a differential operator, acting on the wave function, giving a 1-form corresponding to a ``vector''.  $\mathbf{P}/(-i\hbar)$ acts on the frame basis by natural derivative of wave function with respect to the adiabatic parameter defining a connection over the frame bundle $F$, and $\omega=\mathbf{A}/(i\hbar)$ is the connection matrix.  The connection matrix is frame basis dependent.  A transform of basis will induce a transform of connection matrix which corresponds to the gauge transform physically.  Here the frame is a basis of orthonormal wave functions, and we would only consider a transform to another orthonormal basis, i.e. a \emph{Unitary Transform}.  More specifically, in most cases, we need only consider $SU(N)$ transform.  Attaching this group to every point of the adiabatic=parameter manifold gives a principal bundle $P$, correspondingly the structure group $SU(N)$ also called gauge group in physics.  A local gauge transform is a smooth section of the principal bundle inducing a bundle morphism from $F$ to itself, i.e. a transform over every fibre of $F$.  Physically and mathematically, $\mathbf{A}$ is gauge dependent or frame dependent and it does not transform as a vector but as $\mathbf{A}'=U^\dagger\mathbf{A}U+i\hbar U^\dagger dU$\footnote{A more familiar example to physicists might be the Christoffel in general relativity.}.  Mathematically one can define a homogeneously transforming 2-from associated to the connection
\begin{equation}\label{eq:curvature}
  \Omega=d\omega+\omega\wedge\omega,
\end{equation}
which is called curvature.  Similarly, we can also define a physical quantity corresponding to curvature
\begin{equation}\label{eq:nonAbelianGaugeField}
  \mathbf{B}=d\mathbf{A}+\frac{1}{i\hbar}\mathbf{A}\wedge\mathbf{A},
\end{equation}
which is gauge field\footnote{This gives an intuitive clue to the definition of non-Abelian gauge field. In fact, Yang and Mills~\cite{YangMills} had tried tedious calculations and conter-intuitive tests, and finally find this generalization to non-Abelian gauge field. However by the analog here, this definition is quite nature, because an observable physical quantity should be gauge independent and transform homogeneously.  In simple $U(1)$ case, $\mathbf{A}$ is a 1-form rather than a matrix of 1-form, the second term of $\mathbf{A}\wedge\mathbf{A}$ vanishes naturally due to the antisymmetrical property of wedge and $\mathbf{B}=d\mathbf{A}$ exactly gives the ordinary electromagnetic field.}.

Now let us be back from this long deviation to mathematical jargons.\footnote{Although physicists who are not familiar with differential geometry may consider these strange words of mathematics bizarre and unnecessary, the author find it helpful for both mathematician and physicists to learn the language of each other's.  Therefore we include this short mathematics-physics dictionary here to span a bridge between mathematics and physics.  In this article, we demonstrated calculations by both ordinary physicists using language and mathematician's language.  It would bring us more convenience by using them properly in specific applications.  Readers who need more references in differential geometry may resort to Ref.~\citenum{Taubes} and ~\citenum{Chern}.} For convenience we shall rewrite dark state as
\begin{equation}\label{eq:Darkstate2}
\begin{split}
  |D\rangle=&\frac{\kappa_2|g_1\rangle-\kappa_1|g_2\rangle}
    {\sqrt{{|\kappa_1|}^2+{|\kappa_2|}^2}} \\
           =&\frac{e^{i\phi_2}|g_1\rangle-|\zeta|e^{i\phi_1}|g_2\rangle}
    {\sqrt{1+{|\zeta|}^2}},
\end{split}
\end{equation}
where $\zeta=\kappa_1/\kappa_3=|\zeta|e^i\phi,\ \phi=\phi_1-\phi_2$.  Since now we have no other eigenstates in this degenerate space, we can multiply Eq.~\ref{eq:Darkstate2} by a phase factor $e^{-i\phi_2}$ to make our life easier.  This local gauge transform does not affect curvature, but do changes connection coefficients.  For the consideration of consistency, we shall keep our notation, try to calculate the more complex version.
$\nabla|D\rangle$ is the covariant derivative of the frame basis.  Since we have chosen our basis as the dressed states for every point, then the covariant derivative is naturally defined in this way implied by adiabatic assumption (Physically, it is natural for us to assume $\{|g_i\rangle, |e\rangle\},i=1,2$, states of bare atom, as a flat frame).
\begin{equation}\label{eq:DDark}
  \begin{split}
    \nabla|D\rangle =&|\nabla D\rangle =
    -\frac{|\zeta|d|\zeta|}{1+{|\zeta|}^2}|D\rangle\\  &+\frac{ie^{i\phi_2}d\phi_2|g_1\rangle
    -e^{i\phi_1}(i|\zeta|d\phi_1+d|\zeta|)|g_2\rangle}
    {\sqrt{1+{|\zeta|}^2}}\\
    \langle D|
    =&\frac{e^{-i\phi_2}\langle g_1| -|\zeta|e^{-i\phi_1}| \langle g_2|}
    {\sqrt{1+{|\zeta|}^2}}\\
    \nabla\langle D| =&\langle \nabla D| =
    -\frac{|\zeta|d|\zeta|}{1+{|\zeta|}^2}\langle D|\\  &+\frac{-ie^{-i\phi_2}d\phi_2\langle g_1|
    -e^{-i\phi_1}(-i|\zeta|d\phi_1+d|\zeta|)\langle g_2|}
    {\sqrt{1+{|\zeta|}^2}}.
  \end{split}
\end{equation}
Then we can calculate the vector potential by definition
\begin{equation}\label{eq:vecpotential}
  \begin{split}
    \mathbf{A}
      =&i\hbar \langle D|\nabla D\rangle \\
      =&i\hbar\left(-\frac{|\zeta|d|\zeta|}{1+{|\zeta|}^2}
      +\frac{id\phi_2+|\zeta|(i|\zeta|d\phi_1+ d|\zeta|)} {1+{|\zeta|}^2}\right)\\
      =&-\hbar\frac{d\phi_2+{|\zeta|}^2d\phi_1}{1+{|\zeta|}^2}.
  \end{split}
\end{equation}
This calculation is easy as $\langle D|D\rangle$ and $\langle g_i|g_j\rangle=\delta_{ij}$. For gauge field, we have two method to calculate it.  One is directly, by definition
\begin{equation}\label{eq:AbelianGaugeField}
  \begin{split}
    \mathbf{B}
    =&d\mathbf{A}+\frac{1}{i\hbar}\mathbf{A}\wedge\mathbf{A} \\
    =&-\hbar\left(
    \frac{{2|\zeta|}d|\zeta|\wedge d\phi_1} {1+{|\zeta|}^2}
    -\frac{2|\zeta|d|\zeta|\wedge (d\phi_2+{|\zeta|}^2 d\phi_1)}{ {(1+{|\zeta|}^2)}^2 }\right)\\
    =&-\hbar\left(
    \frac{2{|\zeta|}d|\zeta|\wedge d\phi_1-2|\zeta| d|\zeta|\wedge d\phi_2} {{(1+{|\zeta|}^2)}^2}\right)\\
    =&\hbar\left(
    \frac{( d\phi_1-d\phi_2)\wedge d({|\zeta|}^2)} {{(1+{|\zeta|}^2)}^2}\right)\\
    =&\hbar\frac{d\phi\wedge d({|\zeta|}^2)} {{(1+{|\zeta|}^2)}^2},
  \end{split}
\end{equation}
where we have used $d^2=0$ and $df\wedge dg=-dg\wedge df$, and the second term on the right hand side of first line vanishes for $\mathbf{A}$ is simply a 1-form. Translate back to usual language of vector analysis
\begin{equation}\label{eq:AbelianGaugeField2}
  \begin{split}
    \mathbf{B}
    =&\hbar\frac{\nabla \phi\times \nabla({|\zeta|}^2)} {{(1+{|\zeta|}^2)}^2},
  \end{split}
\end{equation}
which is consistent with Eq.~30 in Ref.~\citenum{DalibardArtificialGauge}. Comparing the calculations before, in this case the language of differential form do save us from the tedious vector product or its components.

We have another way to calculate the gauge field. by definition $\mathbf{B}=\nabla\times\mathbf{A}=i\hbar\nabla\times \langle D|\nabla D\rangle= i\hbar\langle \nabla D|\times|\nabla D\rangle$.  This vector identity is not very obvious, however its differential form correspondence $\langle d D|\wedge |d D \rangle$ is easy to prove (see Theorem~\ref{thm:diffs} in Appendix).  Careful reader must have noticed some strange operations $\langle d *|\wedge |d *\rangle$. In fact it is well defined in a sense that $\langle *||*\rangle$ indicates the Hermitian inner products of Dirac's bra-ket, and $d* \wedge d*$ is normal wedge product. Since the inner product is Hermitian linear integral (sum) over internal space, whereas d is differential operator acting on the $C^\infty$ function over base  manifold, they do commute. Then one can calculate it without any ambiguity. It is more complex to use this formula here to calculate $\mathbf{B}$.  For pedagogical purpose, we demonstrate the calculating rules in the following to convince the reader that it does work (That is why we also calculated $\langle \nabla D|$ in Eq.~\ref{eq:DDark}).
\begin{equation}\label{eq:dddd}
  \begin{split}
    \langle d D|\wedge| d D \rangle
    =&\frac{(-id\phi_2+{|\zeta|}( d|\zeta|-i|\zeta|d\phi_1)) \wedge (-|\zeta|d|\zeta|)}{ {(1+{|\zeta|}^2)}^2 }\\
    &+\frac{ (-|\zeta|d|\zeta|)\wedge(id\phi_2+{|\zeta|}( d|\zeta|+i|\zeta|d\phi_1)) }{ {(1+{|\zeta|}^2)}^2 }\\
    & +\frac{2i|\zeta|d|\zeta|\wedge d\phi_1}{{1+{|\zeta|}^2}}\\
    =& \frac{2i{|\zeta|}(d\phi_2-d\phi_1) \wedge d|\zeta|}{ {(1+{|\zeta|}^2)}^2 }\\
    =& \frac{-i d\phi \wedge d({|\zeta|}^2)}{ {(1+{|\zeta|}^2)}^2 }.
  \end{split}
\end{equation}
At last, as we expected $\mathbf{B}=i\hbar\langle d D|\wedge| d D \rangle$, we cam back to Eq.~\ref{eq:AbelianGaugeField}.

Finally, we sstill need to calculated the vector potential
\[
W=\frac{1}{2M}(\mathbf{A}^2_{\phantom{a}D,D}-{(\mathbf{A}_{D,D})}^2)=
\frac{\sum_{j \neq D}{|\mathbf{A}_{D,j}|}^2}{2M},
\]
where the first term is from the diagonal term of $\mathbf{A}^2$, and the second term is due to the square term of kinetic energy (cf. Eq.~\ref{eq:redKE}).  Since $\mathbf{A}$ is a matrix vector, it carries two kind of indices, indices of Euclidean space and that of internal space. So one should be careful here that we need to deal with two kind of inner products. For the Euclidean one, there is a mathematical correspondence in the language of differential form by using Hodge star operator``$\star$''. The inner product of two vector corresponds to $\omega_A\wedge\star\omega_B$ that is a 3-form (or volume form). For the one of internal space, the inner product is Hermitian.  It will be greatly simplified if make make good use of the orthonormal relations of the bases.
We notice that the summing indices $j$ runs over the subspace orthogonal to $|D\rangle$, if we choose $|B\rangle$ and $|e\rangle$ the basis, there is only one term contributing to the sum ($|e\rangle$ is  orthogonal to $|g_1\rangle$ and $|g_2\rangle$). Therefore
\begin{equation}\label{eq:scalarpot}
  \begin{split}
    W=&\frac{{|\mathbf{A}_{D,B}|}^2}{2M}\\
    =&\hbar^2\frac{{|\nabla (|\zeta|)|}^2+ {|\zeta|}^2{|\nabla \phi|}^2} {2M{(1+{|\zeta|}^2)}^2}.
  \end{split}
\end{equation}
Since the eigen energy of the dark state is zero, no term corresponding to $\hbar\Omega/2$ in Eq.~\ref{eq:redResult} appears here. We at last finish the mountains of calculations.

\begin{figure}
  \centering
  \includegraphics[scale=0.35]{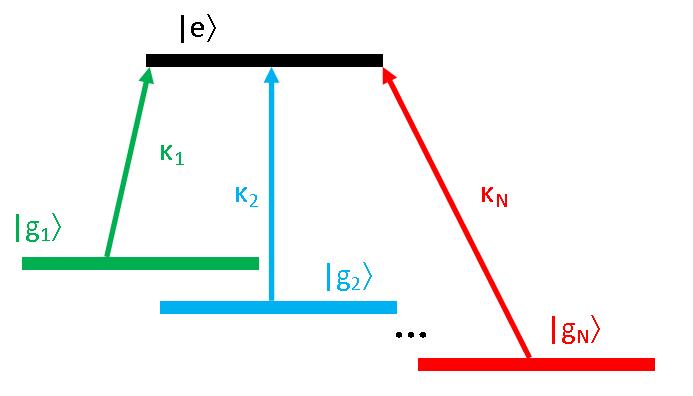}\\
  \caption{N-pod Scheme}\label{fig:npod}
\end{figure}

In experiments, one usually uses cold alkali atom system.  One chooses $|g_i\rangle$ as the hyperfine level of the ground state, and $|e\rangle$ as first excited state.  The fact that no population of dark state is on $|e\rangle$ protects our system from spontaneous emission.  Unfortunately, as we have shown, this lambda scheme still provides us Abelian gauge field.  So as to achieve a non-Abelian scheme, we may simply put more pods on the ground, which naturally leads us to an n-pod system (Fig.~\ref{fig:npod}).

Before we proceed to more general case, let us examine the next simplest case of $N=3$.  All the physical consideration mathematical techniques are exactly the same as $N=1$ of our toy model and $N=2$ of $\Lambda$-scheme, so it is as easy as counting form 1 and 2 to 3.

First, again write down the coupling Hamiltonian of atom and light in RWA (without detuning)~\footnote{For consistency of the symbols in this article our notations are different from Ref.~\citenum{Bergmann1999} that has discussed this scheme in detail. }
\begin{equation}\label{eq:Ham4level}
  H_I=\frac{\hbar}{2}
  \begin{pmatrix}
   0        &    \kappa_1^* &   0         &0\\
   \kappa_1 &    0          &   \kappa_2  &\kappa_3\\
   0        &    \kappa_2^* &   0         &0\\
   0        &    \kappa_3^* &   0         &0\\
  \end{pmatrix}.
\end{equation}
By Dirac notation
\begin{equation*}
  \begin{split}
    H_I=&\frac{\hbar}{2}(\kappa_1|e\rangle\langle g_1|
        +\kappa_2|e\rangle\langle g_2|+\kappa_3|e\rangle\langle g_2|)+h.c. \\
       =& \frac{\hbar}{2}(\kappa|e\rangle\langle B|)+h.c. ,
  \end{split}
\end{equation*}
where $|B\rangle=(\kappa_1^*|g_1\rangle+\kappa_2^*|g_2\rangle+\kappa_3^* |g_3\rangle)/\kappa$, and $\kappa=\sqrt{{|\kappa_1|}^2+{|\kappa_2|}^2+{|\kappa_3|}^2}$. This time the interaction still only couples $|e\rangle$ and $|B\rangle$ together, but now we have a Hilbert space of dimension 4!  So there are another two degenerate dark states of eigen energies $0$.
\begin{equation}\label{eq:Darkstatetripod}
  \begin{split}
    |D_1\rangle=&(\kappa_2|g_1\rangle-\kappa_1|g_2\rangle)/\eta,\\
    |D_2\rangle=&(\kappa_3\frac{\kappa_1^*}{\eta}|g_1\rangle
    +\kappa_3\frac{\kappa_2^*}{\eta}|g_2\rangle-\eta|g_3\rangle)
    /\kappa,
  \end{split}
\end{equation}
where $\eta={({|\kappa_1|}^2+{|\kappa_2|}^2)}^{{1}/{2}}$, who are orthogonal and uncoupled to $|e\rangle$ and $|B\rangle$.  On the other hand, $|e\rangle$ and $|B\rangle$ split into $|\pm\rangle=(|e\rangle \pm |B\rangle)/\sqrt{2}$ with eigenenergies $E_\pm=\pm\hbar\kappa/2$. It is more convenient to define several angle parameters
\begin{equation}\label{eq:Darkangle}
  \begin{split}
   \sin \vartheta=&|\kappa_1|/\eta,
   \cos \vartheta=|\kappa_2|/\eta,\\
   \sin \varphi=&|\kappa_3|/\kappa,
   \cos \varphi=\eta/\kappa.
  \end{split}
\end{equation}
Denote $\kappa_j=|\kappa_j|e^{iS_j}$\footnote{To avoid ambiguity of the symbols, $\varphi$ and $\vartheta$ are angle parameters, and $S_j$ are used for phase instead of $\phi$ in $\Lambda$-scheme.}.
Then, under the basis of $<|g_1\rangle,|e\rangle,|g_2\rangle,|g_3\rangle> $, four eigenstates of the coupling Hamiltonian become
\begin{widetext}
\begin{equation}\label{eq:Darkstatematrix}
  \begin{split}
  |D_1\rangle=&
    \begin{pmatrix}
    \cos \vartheta e^{iS_2}\\ 0\\ -\sin \vartheta e^{iS_1}\\ 0
    \end{pmatrix},
  |D_2\rangle=
    \begin{pmatrix}
    \sin \varphi {\sin \vartheta}e^{i(S_3-S_1)}\\ 0\\ \sin \varphi  {\cos \vartheta}e^{i(S_3-S_2)}\\ -\cos \varphi
    \end{pmatrix},
  |+\rangle=\frac{1}{\sqrt{2}}
    \begin{pmatrix}
    \cos \varphi {\sin \vartheta}e^{-iS_1}\\ 1\\ \cos \varphi  {\cos \vartheta}e^{-iS_2}\\ \sin \varphi e^{-iS_3}
    \end{pmatrix},
  |-\rangle=\frac{-1}{\sqrt{2}}
    \begin{pmatrix}
    \cos \varphi {\sin \vartheta}e^{-iS_1}\\ -1\\ \cos \varphi  {\cos \vartheta}e^{-iS_2}\\ \sin \varphi e^{-iS_3}
    \end{pmatrix}.\\
    &\phantom{=}
  \end{split}
\end{equation}
\end{widetext}

With this new group of basis, redo
the procedures before.  We have the projection as
\begin{equation}\label{eq:NDproj}
  |\Psi (\mathbf{r},t)\rangle = \sum_{j=D_k,\pm}\psi_j (\mathbf{r},t)|\chi_j\rangle,
\end{equation}
where $k=1,2$ for tripod scheme (for $N$-pod scheme, with total $N+1$ levels, $k=1,2,3...N-1$ counting the dart states).  Write down the time dependent Schr\"odinger equation in terms of the basis $<|D_k\rangle,k=1,2...N-1,|\pm\rangle>$:
\begin{widetext}
\begin{equation}\label{eq:generaltdSchroedinger}
  \begin{split}
     &i\hbar \frac{\partial}{\partial t}
       \begin{pmatrix}\psi_D(\mathbf{r},t)\\ \psi_B(\mathbf{r},t) \end{pmatrix}
     =\frac{{(\mathbf{p}\hat{\mathbb{I}}
         -\hat{\mathbf{A}})}^2}{2M}
         \begin{pmatrix}\psi_D(\mathbf{r},t)\\ \psi_B(\mathbf{r},t) \end{pmatrix}
      + \begin{pmatrix}V_D(\mathbf{r})+U_D& 0\\0 & V_B(\mathbf{r})+U_B \end{pmatrix}
      \begin{pmatrix}\psi_D(\mathbf{r},t)\\ \psi_B(\mathbf{r},t) \end{pmatrix}\\
      =&\left[\frac{1}{2M}
      \begin{pmatrix}
         {(\mathbf{p}\hat{\mathbb{I}}
         -\hat{\mathbf{A}}_{D,D})}^2& 0\\
         0 &{(\mathbf{p}\hat{\mathbb{I}}
         -\hat{\mathbf{A}}_{B,B})}^2
      \end{pmatrix}
      + \frac{1}{2M}\begin{pmatrix}
      \hat{\mathbf{A}}_{D,B}\cdot\hat{\mathbf{A}}_{B,D}& \sim0\\
      \sim0 & \hat{\mathbf{A}}_{B,D}\cdot\hat{\mathbf{A}}_{D,B}
       \end{pmatrix}
      + \begin{pmatrix}\tilde{V}_D(\mathbf{r})& 0\\0 & \tilde{V}_B(\mathbf{r}) \end{pmatrix}
      \right]
      \begin{pmatrix}\psi_D(\mathbf{r},t)\\ \psi_B(\mathbf{r},t) \end{pmatrix}\\
      =&\left[
      \begin{pmatrix}
         \frac{1}{2M}{(\mathbf{p}\hat{\mathbb{I}}
         -\hat{\mathbf{A}}_{D,D})}^2+\tilde{V}_D(\mathbf{r})+{W}_D(\mathbf{r})& 0\\
         0 &\frac{1}{2M}{(\mathbf{p}\hat{\mathbb{I}}
         -\hat{\mathbf{A}}_{B,B})}^2+\tilde{V}_B(\mathbf{r})+{W}_B(\mathbf{r})
      \end{pmatrix}
      \right]
      \begin{pmatrix}\psi_D(\mathbf{r},t)\\ \psi_B(\mathbf{r},t) \end{pmatrix},
  \end{split}
\end{equation}
\end{widetext}
where $\hat{\mathbf{A}}_{J,K}=i\hbar\langle J|\nabla |K \rangle, J,K=D,B$ and suffices $D=$dark states submanifold and $B=$bright states submanifold.  Since we used basis of eigenstates of atom-light coupling, $U$ is diagonal.  Since the dark states and bright states are separated by a finite gap, by adiabatic condition, the off-diagonal entries of  $\hat{\mathbf{A}}\cdot \mathbf{p}$ and $\hat{\mathbf{A}}^2$ are higher order terms which are neglected.  Therefore the motions of bright and dark submanifolds are decoupled. The diagonal terms of  $\hat{\mathbf{A}}^2$ ( $\hat{\mathbf{A}}^2_{\phantom{*}J,J}$ are absorbed into the the square terms) contributes a scalar potential $W$ that is quite different from $V$.  $W$ is not diagonal in general; it is as $\mathbf{A}$ a non-Abelian potential coupling the components of wave function in the submanifold.  Since the excited state is orthogonal to dark submanifold, dark states do not suffer form spontaneous emission, which provides an practical model for artificial non-Abelian gauge fields.  By projecting to this submanifold, we arrive at last to our effective EoM with non-Abelian gauge fields

\begin{equation}\label{eq:generaleffEoM}
  \begin{split}
     i\hbar \frac{\partial}{\partial t}
       \psi_D(\mathbf{r},t)
     =\left[\frac{1}{2M}{(\mathbf{p}\hat{\mathbb{I}}
         -\hat{\mathbf{A}}_{D})}^2+\tilde{V}(\mathbf{r})+{W}(\mathbf{r})
      \right]\psi_D(\mathbf{r},t),
  \end{split}
\end{equation}
where non-Abelian scalar potential
\begin{equation}\label{eq:NonAbelScalarPot}
  \begin{split}
W=
\frac{1}{2M}\sum_{B (\neq D)}{\mathbf{A}_{D,B}\cdot\mathbf{A}_{B,D}}.
  \end{split}
\end{equation}
\begin{widetext}
The explicit results are calculated\footnote{These results are calculated by Ruseckas {\sl et. al} (Ref.~\citenum{Ruseckas2005}) as well as Unanyan, Shore and Bergmann (Ref.~\citenum{Bergmann1999}).  The convention of the base vectors has a little difference here.  We chose a different angle parameter and phase factor.  The result in Ref.~\citenum{Ruseckas2005} can be obtained by a non-Abelian gauge transform $\mathbf{A}'=U^\dagger\mathbf{A}U+i\hbar U^\dagger dU$. This result can also be found in review article Ref.~\citenum{DalibardArtificialGauge}.}, vector potential
\begin{equation}\label{eq:3podNonAbelvecPot}
  \begin{split}
  \mathbf{A}_{11}=&-\hbar(\cos^2 \vartheta dS_2+\sin^2 \vartheta dS_1),\\
  \mathbf{A}_{12}=&\hbar\sin \varphi e^{i(S_3-S_1-S_2)}\Bigg[\sin \vartheta \cos\vartheta (dS_1-dS_2)+id\vartheta \Bigg],\\
  \mathbf{A}_{22}=&\hbar\sin^2 \varphi\Bigg[\sin^2 \vartheta (dS_1-dS_3)+\cos^2 \vartheta (dS_2-dS_3)\Bigg],\\
  \mathbf{A}_{1\pm}=&\pm\frac{\hbar}{\sqrt{2}}\cos \varphi e^{-i(S_1+S_2)}\Bigg[\sin \vartheta \cos\vartheta (dS_1-dS_2)+id\vartheta \Bigg],\\
  \mathbf{A}_{2\pm}=&\pm\frac{\hbar}{\sqrt{2}} e^{-iS_3}\Bigg[\sin \varphi \cos\varphi (\sin^2\vartheta dS_1+ \cos^2\vartheta dS_2  - dS_3)-id\varphi \Bigg],
  \end{split}
\end{equation}
and scalar potential
\begin{equation}\label{eq:3podNonAbelscalPot}
  \begin{split}
  W_{11}=&\frac{{|\mathbf{A}_{1+}|}^2+{|\mathbf{A}_{1-}|}^2}{2M}
  =\frac{\hbar^2}{2M}\cos^2 \varphi\Bigg[\sin^2 \vartheta \cos^2\vartheta {|\nabla(S_1-S_2)|}^2+{|\nabla\vartheta|}^2 \Bigg],\\
  W_{12}=&\frac{{\mathbf{A}_{1+}}\cdot{\mathbf{A}_{+2}}+{\mathbf{A}_{1-}} \cdot{\mathbf{A}_{-2}}}{2M}\\
  =&\frac{\hbar^2}{2M}\cos \varphi e^{i(S_3-S_1-S_2)}\Bigg[\frac{\sin 2\vartheta}{2} \nabla(S_1-S_2)+i\nabla\vartheta \Bigg]
  \cdot \Bigg[\frac{\sin 2\varphi}{2} (\sin^2\vartheta \nabla S_1+ \cos^2\vartheta \nabla S_2 - \nabla S_3)+i\nabla \varphi \Bigg]\\
  W_{22}=&\frac{{|\mathbf{A}_{2+}|}^2+{|\mathbf{A}_{2-}|}^2}{2M}
  =\frac{\hbar}{2M}\Bigg[\sin^2 \varphi \cos^2 \varphi {|\sin^2\vartheta \nabla S_1+ \cos^2\vartheta \nabla S_2} {- \nabla S_3|}^2 +{|\nabla\varphi|}^2 \Bigg].
  \end{split}
\end{equation}
For convenience (it is also practical in experiment as long as the laser coupling $|1\rangle$ and $|2\rangle$ are copropagating and have the same frequency as well as orbital angular momentum), we may set $S_1=S_2=S$.   By adjusting the relative phase of the third laser we may $S_3=0$.  This leads to
\begin{equation}\label{eq:3podNonAbelvecPot2}
  \begin{split}
  {{\mathbf{A}}}=&\hbar
\begin{pmatrix}
  -dS &
  i \sin \varphi e^{-2iS} d\vartheta\\
  -i \sin \varphi e^{2iS} d\vartheta
  & \sin^2 \varphi dS
\end{pmatrix}.\\
  \end{split}
\end{equation}
We further to calculate the nun-Abelian gauge field.  First calculate the exterior differential of $\mathbf{A}$
\begin{equation}\label{eq:diffNonAbelvecPot}
  \begin{split}
  d{{\mathbf{A}}}=&\hbar
\begin{pmatrix}
  0 &
  i \cos \varphi e^{-2iS} d\varphi\wedge d\vartheta
  +i \sin \varphi e^{-2iS} (-2i)dS\wedge d\vartheta\\
  -i \cos \varphi e^{2iS} d\varphi\wedge d\vartheta
  -i \sin \varphi e^{2iS}   (2i)dS\wedge d\vartheta
  & 2\sin \varphi \cos \varphi d\varphi \wedge dS
\end{pmatrix}\\
  =&\hbar
\begin{pmatrix}
  0 &
  e^{-2iS} (i \cos \varphi d\varphi\wedge d\vartheta
   +\sin \varphi 2dS\wedge d\vartheta)\\
   e^{2iS} (-i \cos \varphi d\varphi\wedge d\vartheta
   +\sin \varphi 2dS\wedge d\vartheta)
  & 2\sin \varphi \cos \varphi d\varphi \wedge dS
\end{pmatrix},\\
  \end{split}
\end{equation}
Then we calculate the exterior product of $\mathbf{A}$
\begin{equation}\label{eq:extprodA}
  \begin{split}
  {{\mathbf{A}}\wedge{\mathbf{A}}}=&\hbar^2
\begin{pmatrix}
  -dS &
  i \sin \varphi e^{-2iS} d\vartheta\\
  -i \sin \varphi e^{2iS} d\vartheta
  & \sin^2 \varphi dS
\end{pmatrix}\wedge
\begin{pmatrix}
  -dS &
  i \sin \varphi e^{-2iS} d\vartheta\\
  -i \sin \varphi e^{2iS} d\vartheta
  & \sin^2 \varphi dS
\end{pmatrix}\\
  =&\hbar^2
  \begin{pmatrix}
  0 &
  - dS\wedge i \sin \varphi e^{-2iS} d\vartheta+ i \sin \varphi e^{-2iS} d\vartheta\wedge\sin^2 \varphi dS\\
  i \sin \varphi e^{2iS} d\vartheta \wedge dS -\sin^2 \varphi dS \wedge i \sin \varphi e^{2iS} d\vartheta
  & 0
  \end{pmatrix}\\
  =&\hbar^2
  \begin{pmatrix}
  0 &
  -i\sin \varphi e^{-2iS}(1+\sin^2\varphi ) dS \wedge d\vartheta\\
  -i \sin \varphi e^{2iS}(1+\sin^2 \varphi) dS \wedge d\vartheta
  & 0
  \end{pmatrix},
  \end{split}
\end{equation}
Combining Eq.~\ref{eq:diffNonAbelvecPot} and ~\ref{eq:extprodA}, we finally obtain the non-Abelian gauge field\footnote{This result was also given in Ref.~\citenum{Ruseckas2005}. The entry of $\mathbf{B}_{12}$ in Eq.(17) in Ref.~\citenum{Ruseckas2005} was wrong.  Our convention here is different from the one in Ref.~\citenum{Ruseckas2005}.  For comparison, one may substitute $(\pi/2-\phi)$ for $\vartheta$ and   $(\pi/2-\theta)$ for $\varphi$ , and do a unitary transform of $U=[e^{-2iS}, 0;0, 1]$, which leads to the correct result to Eq.(17) in Ref.~\citenum{Ruseckas2005}.}
\begin{equation}\label{eq:NonAbelianField}
  \begin{split}
\mathbf{B}
    =&d\mathbf{A}+\frac{1}{i\hbar}\mathbf{A}\wedge\mathbf{A}\\
  =&\hbar
\begin{pmatrix}
  0 &
  e^{-2iS} (i \cos \varphi d\varphi\wedge d\vartheta
   +\sin \varphi 2dS\wedge d\vartheta)\\
   e^{2iS} (-i \cos \varphi d\varphi\wedge d\vartheta
   +\sin \varphi 2dS\wedge d\vartheta)
  & 2\sin \varphi \cos \varphi d\varphi \wedge dS
\end{pmatrix}\\
&+\hbar
  \begin{pmatrix}
  0 &
  -\sin \varphi e^{-2iS}(1+\sin^2\varphi ) dS \wedge d\vartheta\\
  - \sin \varphi e^{2iS}(1+\sin^2 \varphi) dS \wedge d\vartheta
  & 0
  \end{pmatrix}\\
  =&\hbar
\begin{pmatrix}
  0 &
  e^{-2iS} (i \cos \varphi d\varphi\wedge d\vartheta
   +\sin \varphi (1-\sin^2 \varphi)dS\wedge d\vartheta)\\
   e^{2iS} (-i \cos \varphi d\varphi\wedge d\vartheta
   +\sin \varphi (1-\sin^2 \varphi)dS\wedge d\vartheta)
  & 2\sin \varphi \cos \varphi d\varphi \wedge dS
\end{pmatrix}.
  \end{split}
\end{equation}
\end{widetext}

\subsection{\label{sec:RotGas}Rotating Gases}
In the previous sections, we have constructed artificial gauge fields through Berry phase induced by atom-light coupling.  If we consider the simplest case of $U(1)$ electromagnetic field,  it acts on the particle by a Lorentz force.  A classical analogy to the Lorentz force in Coriolis force in rotating frame.  A typical example of geometrical phase in classical mechanics is the Foucault's pendulum.  This analogy encourages us to generate the artificial magnetic in a rotating gas system, which is, in fact, realized earlier historically than the optical methods~\cite{Cornell1999}~\cite{Madison}~\cite{Shaeer2001}.

Consider a 2-dimensional gas disc (the atoms are confined in x-y plane ). Its Hamiltonian is
\begin{equation}\label{eq:HamAtomRest}
  H=\frac{\hat{\mathbf{p}}^2}{2m}+V(\hat{\mathbf{r}}),
\end{equation}
where $\hat{\mathbf{p}}$ and $\hat{\mathbf{r}}$ are operators in the rest frame.
Now let us rotate our system around z-axis with a constant angular velocity $\mathbf{\Omega}_{rot}=\Omega_{rot} \mathbf{e}_z$.  In quantum mechanics, rotation generator of coordinate is the angular momentum operator $\hat{\mathbf{L}}$.  A general rotation around $\mathbf{n}$ by angle $\theta$ can be written in terms of generator by exponential map  $\hat{\mathcal{R}}_\mathbf{n}(\theta)=e^{-\frac{i}{\hbar}\theta\mathbf{n}\cdot\hat{\mathbf{L}}}$, where $\mathbf{n}$ is a unit vector indicating the rotation axis.  So we may rotate our system by operator
\begin{equation}\label{eq:RotOper}
  \hat{\mathcal{R}}_z(\Omega_{rot} t) =e^{-\frac{i}{\hbar}\Omega_{rot} t\hat{L}_z},
\end{equation}
where $\hat{L}_z$ is the z-component of angular momentum.  Correspondingly,  our Hamiltonian should transform as
\begin{equation}\label{eq:HamAtomRestTD}
  \hat{H}(t)=\hat{\mathcal{R}}_z(\Omega_{rot} t)H\hat{\mathcal{R}}^\dagger_z(\Omega_{rot} t),
\end{equation}
where $\hat{\mathbf{r}}$ transform to $\hat{\mathbf{r}}'=\hat{\mathcal{R}}_z(\Omega_{rot} t)\hat{\mathbf{r}}\hat{\mathcal{R}}^\dagger_z(\Omega_{rot} t)$, and $\hat{\mathbf{p}}^2$ is invariant under rotation $\hat{\mathbf{p}}'^2=\hat{\mathcal{R}}_z(\Omega_{rot} t)\hat{\mathbf{p}}^2\hat{\mathcal{R}}^\dagger_z(\Omega_{rot} t)=\hat{\mathbf{p}}^2$ \footnote{One can verify that $[\hat{\mathbf{p}}^2,L_z]=0$ by brute force calculation.  More intuitively, since $\hat{\mathbf{p}}^2$ is an Euclidean scalar, it is of course in variant under a rotation.}.
Then we have our time-dependent Schr\"odinger equation (TDSE)
\begin{equation}\label{eq:TDSE}
  i\hbar\frac{\partial}{\partial t} |\psi\rangle =\hat{H}(t)|\psi\rangle,
\end{equation}
where $|\psi\rangle$ is the state vector in rotating frame.  Since we stay in the lab frame, we need to transform the state vector back to lab frame $|\psi'\rangle=\hat{\mathcal{R}}^{-1}_z(\Omega_{rot} t)|\psi\rangle=\hat{\mathcal{R}}^\dagger_z(\Omega_{rot} t)|\psi\rangle$, which leads to the TDSE in lab frame
\begin{equation}\label{eq:TDSELab}
\begin{split}
  i\hbar\frac{\partial}{\partial t} |\psi'\rangle=&i\hbar\frac{\partial}{\partial t} \left( \hat{\mathcal{R}}^\dagger_z(\Omega_{rot} t)|\psi\rangle\right)\\
  =&i\hbar\frac{\partial}{\partial t} \left( \hat{\mathcal{R}}^\dagger_z(\Omega_{rot} t)\right)|\psi\rangle
  +\hat{\mathcal{R}}^\dagger_z(\Omega_{rot} t)\left[i\hbar\frac{\partial}{\partial t}  |\psi\rangle \right]\\
  =&i\hbar\frac{\partial}{\partial t} \left( \hat{\mathcal{R}}^\dagger_z(\Omega_{rot} t)\right) \hat{\mathcal{R}}_z(\Omega_{rot} t)
  \hat{\mathcal{R}}^\dagger_z(\Omega_{rot} t)|\psi\rangle\\
  &+\hat{\mathcal{R}}^\dagger_z(\Omega_{rot} t)\hat{H}(t)   \hat{\mathcal{R}}_z(\Omega_{rot} t)
  \hat{\mathcal{R}}^\dagger_z(\Omega_{rot} t)|\psi\rangle \\
  =&- \Omega_{rot} \hat{L}_z |\psi'\rangle +{H}|\psi'\rangle =H'|\psi'\rangle,
\end{split}
\end{equation}
where
\begin{equation}\label{eq:effectiveHam}
  H'=H-\mathbf{\Omega}_{rot}\cdot\hat{\mathbf{L}},
\end{equation}
is the effective Hamiltonian in lab frame that is time independent.
Noting that $ \mathbf{\Omega}_{rot}\cdot\hat{\mathbf{L}}= \mathbf{\Omega}_{rot}\cdot (\hat{\mathbf{r}}\times\hat{\mathbf{p}}) = (\mathbf{\Omega}_{rot}\times \hat{\mathbf{r}})\cdot\hat{\mathbf{p}}$ and $[(\mathbf{\Omega}_{rot}\times \hat{\mathbf{r}})_a,\hat{\mathbf{p}}_a]=0$, we can rewrite the effective Hamiltonian as
\begin{equation}\label{eq:effectiveHam2}
\begin{split}
  H'=&\frac{\hat{\mathbf{p}}^2}{2m}+V(\hat{\mathbf{r}})-\mathbf{\Omega}_{rot}\cdot\hat{\mathbf{L}}\\
  =&\frac{{(\hat{\mathbf{p}}-\hat{\mathbf{A}})}^2}{2m}+V(\hat{\mathbf{r}}) +W_{rot}(\hat{\mathbf{r}}),
\end{split}
\end{equation}
where
\begin{equation}\label{eq:vectPot}
  \hat{\mathbf{A}}=m\mathbf{\Omega}_{rot}\times \hat{\mathbf{r}},
\end{equation}
is the effective vector potential, and
\begin{equation}\label{eq:scalarPot}
  W_{rot}(\hat{\mathbf{r}})=-\frac{1}{2}m{(\mathbf{\Omega}_{rot}\times \hat{\mathbf{r}})}^2,
\end{equation}
is the centrifugal potential.  In the specific configuration that we consider here, $\hat{\mathbf{r}}$ is in the x-y plane and $\mathbf{\Omega}_{rot}$ is along z-axis.  Thus this term reduces to
\begin{equation}\label{eq:scalarPot2}
  W_{rot}(\hat{\mathbf{r}})=-\frac{1}{2}m{\Omega}_{rot}^2\hat{\mathbf{r}}^2.
\end{equation}
The induced effective magnetic field is
\begin{equation}\label{eq:magfield}
  \hat{\mathbf{B}}=\nabla\times\hat{\mathbf{A}}=2m{\Omega}_{rot}\mathbf{e}_z.
\end{equation}

If our system in x-y plane is confined by a harmonic potential
\begin{equation}\label{eq:harmPot}
  V(\hat{\mathbf{r}})=\frac{1}{2}m(\omega_x^2\hat{{x}}^2+\omega_y^2\hat{{y}}^2),
\end{equation}
where $\omega_x, \omega_y\sim \omega$ characterize the confinement strength.  When the angular velocity of rotation $\Omega_{rot}$ approaches $\omega$, the centrifugal anti-trapping potential compensates the confinement so that the atoms move ``freely'' only seeing the magnetic field.  The associated cyclotron frequency is $\omega_c=B/m=2\Omega_{rot}$.

The Hamiltonian shown above is the single-particle Hamiltonian.  We can further take the two-body interaction in to account without changing much in the mathematical frame.  All that one need to observe is that $\mathbf{r}_1\cdot\mathbf{r}_2$ is invariant under rotation as it is Euclidean scalar, which leads the invariance of ${|\mathbf{r}_1-\mathbf{r}_2|}^2=\mathbf{r}_1^2+\mathbf{r}_2^2-2\mathbf{r}_1\cdot\mathbf{r}_2$.  If two-body interaction $V_12$ depends only on $|\mathbf{r}_1-\mathbf{r}_2|$, all the mathematics above are still valid, therefore this formalism can be generalized to interacting particles.

The early experiments are demonstrated by the groups at JILA~\cite{Cornell1999}, ENS~\cite{Madison} and MIT~\cite{Shaeer2001}.
Creating vortex in BEC was originally reported in Ref.~\citenum{Cornell1999} by rotating through atom-light coupling of the internal degree of freedom.  In Ref.~\citenum{Madison}, they rotate the BEC by stirring  with an an optical spoon.  Several vortices were observed, which are imaged by Time of Flight (ToF). Ref.~\citenum{Shaeer2001} used dipole force exerted by two blue-detuned laser beams.  The ordered vortex lattices were  observed by resonant absorption imaging.  Local structure, defects, long range order and finite size effects were studied.
Further experiments to achieve rapid rotating BEC adopted the deformed trap method cooperating with evaporative~\cite{Cornell2003} and optical spin-up techniques~\cite{Cornell2004}.  Condensation (of Boson) near the Lowest Landau Level (LLL) was achieved.  In the regime of low filling factor $\nu\lesssim 10$, the BEC would exhibit strongly correlated characters, where melting of vortex lattice due to quantum fluctuation, emergence of fractional statistical quasiparticle excitations and other interesting features may appear.

\section{Spin-Orbit Coupling}
\subsection{Generating Spin-Orbit Coupling in Cold Atom System}
In non-relativistic quantum mechanics, the spin-orbit coupling induced from the non-relativistic limit of Dirac equation.  It emerges as a magnetic field seen by the moving electron (Lorentz transform of electromagnetic field) coupling the spin in dipole form (to the leading order of non-relativistic limit).  The magnetic field depends on the motion of the electron.  If we study the motion in terms of Hamiltonian, the electromagnetic coupling enters through vector potential $\mathbf{A}$ and scalar potential $V$.  If  $\mathbf{A}$ is constant, $\mathbf{B}=\nabla\times\mathbf{A}=0$, therefore no spin-orbital coupling.  This feature is no longer valid in non-Abelian case.  Even if $\mathbf{A}$ is constant, $\mathbf{B}=\nabla\times\mathbf{A}+\frac{1}{i\hbar}\mathbf{A}\times\mathbf{A}=\frac{1}{i\hbar}\mathbf{A}\times\mathbf{A}$ does not vanish in general.  Secondly, the off diagonal terms in $\mathbf{A}$ coupling spin (or pseudo-spin) directly through kinetic term of Hamiltonian
\begin{equation}\label{eq:NonAbelianHamKin}
  H_{kin}=\frac{{(\mathbf{p}-\mathbf{A})}^2}{2m}.
\end{equation}
Therefore non-Abelian gauge fields inevitably leads to spin-orbit coupling.  For example, the non-Abelian gauge field may invert isospin of the nucleon scattering a proton to a neutron~\cite{WuYang}~\cite{Horvathy}.  The non-Abelian gauge field also leads to new characters in A-B effects~\cite{Horvathy}~\cite{Jacob}.

In condensed matter physic, Spin-Orbit Coupling (SOC) plays important roles in spin Hall effect~\cite{Hirsch}, anomalous Hall effect~\cite{AnomalousHall}, quantum spin Hall effect~\cite{QSpinHall}~\cite{Z2SPH} or topological insulator~\cite{TIin3D}~\cite{HasaneKane}~\cite{QiZhang}, and quantum anomalous Hall effect~\cite{QAHE}.

It is instructive to construct the non-Abelian gauge fields and to induce SOC in cold atom systems.  We may control the system by finely tuning the laser and simulate those complex system.  Effective field in cold atom systems provide a another way to study some high energy physics and particle physics, e.g. searching for Majorana modes, on the table without large accelerator.  On the other hand cold atom also offers a setting to simulate the interacting may-body systems.  In real materials it is electron, a Fermion who plays the central role, although although quasiparticle excitations of Boson may emerges.  Cold atoms can be Fermion as well as Boson, which has its unique advantages.

\begin{figure}
  \centering
  \includegraphics[scale=0.4]{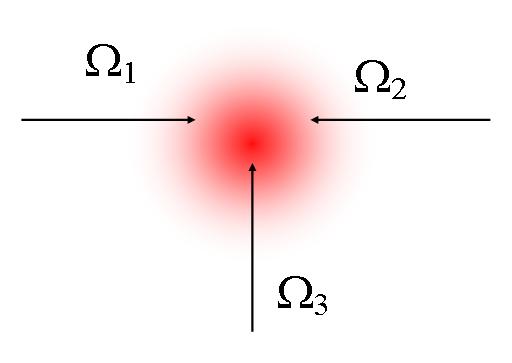}\\
  \caption{Laser Configuration for Generating Spin-Orbit Coupling in Tripod Scheme}\label{fig:SOCTripod.jpg}
\end{figure}

Several schemes were proposed to generating SOC in cold atom systems.  One of the intuitive way follows from our tripod model.  Since $\mathbf{A}$ is a Hermitian $2\times2$ matrix, what would happen if it is proportional to the spin operator?
A scheme  proposed in Ref.~\citenum{Ohberg2008a} and ~\citenum{Ohberg2008b} uses degenerate dark states in tripod model whose laser field configuration is shown in Fig.~\ref{fig:SOCTripod.jpg}.  Properly choosing laser such that $\mathbf{A}=-\hbar\kappa\bm{\sigma}_\perp$ , the effective Hamiltonian reduce to
\begin{equation}\label{eq:effHamSOC}
  H=\frac{{(-i\hbar\nabla+\hbar\kappa\bm{\sigma}_\perp)}^2}{2m}+V,
\end{equation}
where $\bm{\sigma}_\perp=\sigma_x\mathbf{e}_x+\sigma_y\mathbf{e}_y$ and $V$ is the total potential including trapping and effective potential that can be tuned diagonal.  Suppose $V$ is constant, therefore can be dropped out temporarily.  Hamiltonian acting on the plane wave solution $\bm{\psi}_{\mathbf{k}}=\psi e^{i\mathbf{k}\cdot\mathbf{r}}$, we obtain
\begin{equation}\label{eq:effHamSOC2}
  H_{\mathbf{k}}=\frac{\hbar^2}{2m} {(\mathbf{k}+\kappa\bm{\sigma}_\perp)}^2.
\end{equation}
It has two branches corresponding to two eigenspinors. The dispersion relation is shown in Fig.~\ref{fig:dispSOC} in terms of natural energy scale $v_0=\hbar\kappa/m$.
\begin{figure}
    \centering
    \includegraphics[scale=0.06]{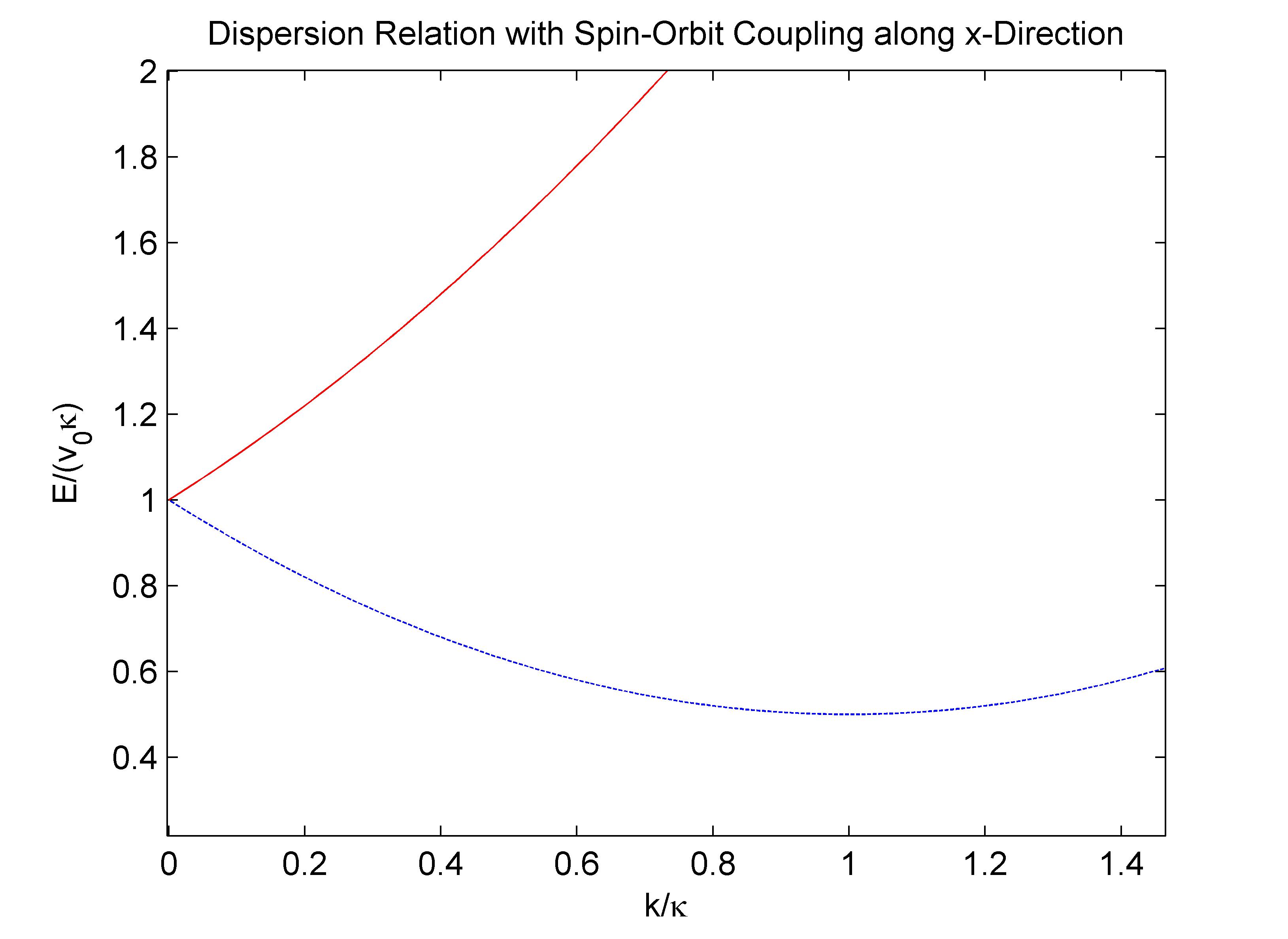}
    \includegraphics[scale=0.06]{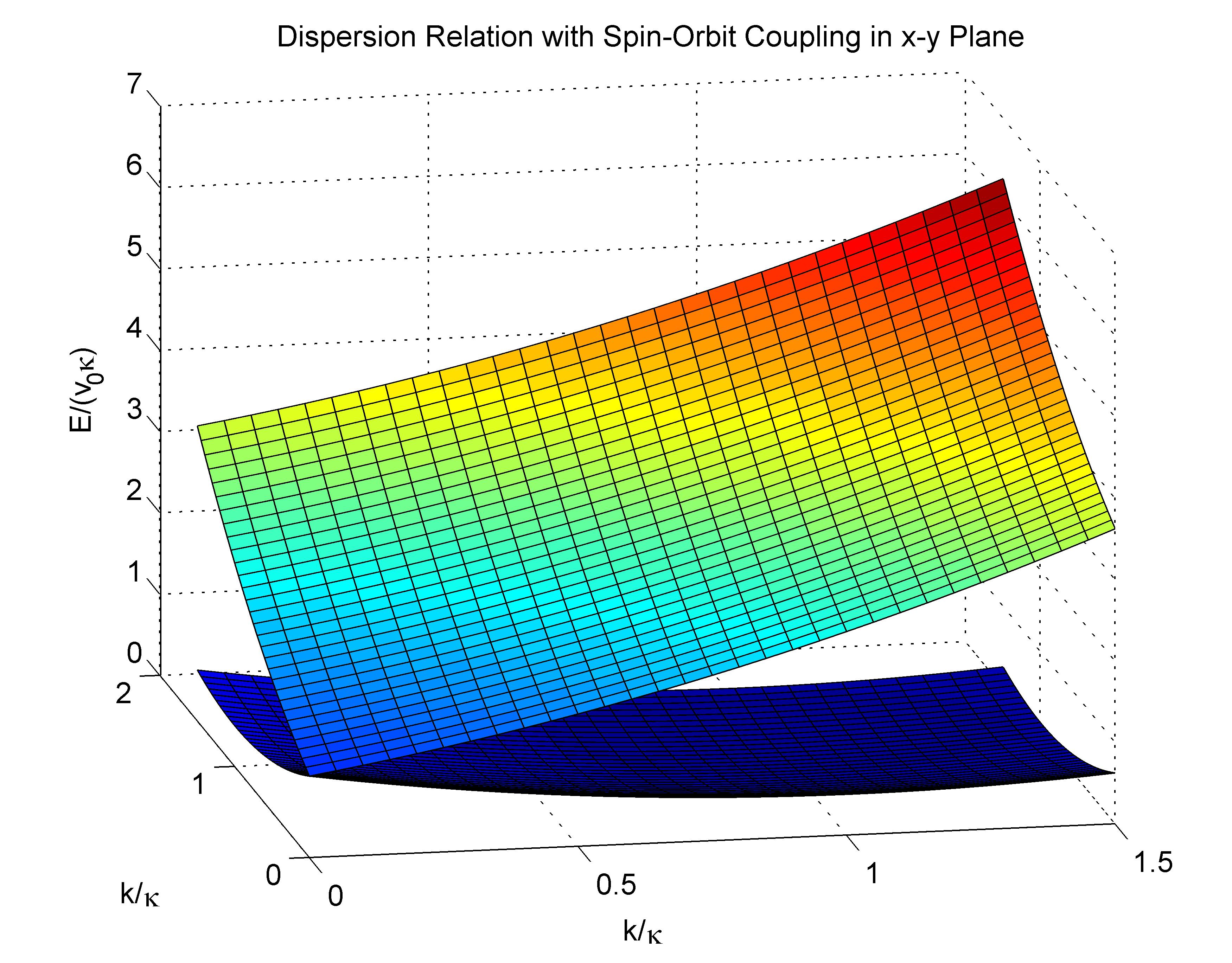}\\
  \caption{Dispersion Relation with Spin-Orbit Coupling}\label{fig:dispSOC}
\end{figure}
Two branches touch at the origin.  Expand $H_\mathbf{k}$ near the origin and neglect the higher order term of $\mathbf{k}\rightarrow 0$, the Hamiltonian becomes
\begin{equation}\label{eq:linHam}
  H_\mathbf{k}=\hbar v_0\mathbf{k}\cdot\bm{\sigma}_\perp+\frac{1}{2}mv_0^2,
\end{equation}
which is a 2-dimensional Dirac Hamiltonian.  This low energy effective Hamiltonian masters a massless particle obeying the Weyl equation.  $v_0$ plays the role of speed of light, whereas, it is recoil velocity of the typical order $1cm/s$ associated with wave vector $\kappa$.  The Hamiltonian $H_\mathbf{k}$ commutes with 2D chirality operator $\sigma=\mathbf{k}\cdot\bm{\sigma}_\perp/k$.  It was predicted to be able to observe negative reflection and Veselago-type lenses~\cite{Ohberg2008a}~\cite{Ohberg2008b}.  A similar effective Hamiltonian emerges from Dirac cones in graphene~\cite{Wallace1947}~\cite{McClure}.  Non-Abelian AB effect was also proposed in a same configuration~\cite{Jacob}.

Another scheme to generate generalized Rashba-Dresselhaus SOC with regular polygon configuration of lasers was proposed by Juzeli\ifmmode \bar{u}\else \={u}\fi{}nas, Ruseckas and Dalibard~\cite{Dalibard2010DRSOC}.

As we have shown in section~\ref{sec:NonAbelian}, although dark states has no population in the excited state, it is not the lowest dressed level.  Therefore collisions can still scatter the atoms out of dark submanifold to the ground dressed state.  To over come this difficulty, another scheme cyclically coupling N ground or metastable states was proposed by Campbell, Juzeli\ifmmode \bar{u}\else \={u}\fi{}nas and Spielman~\cite{Campbell}.

\section{Optical Flux Lattice}
\subsection{Optical Lattice}
The confinement of cold atoms can be achieved by utilizing the dipole force~\cite{Grimm200095}
\begin{equation}\label{eq:dipoleforce}
  \mathbf{F}_{dip}=-\nabla U_{dip}=\frac{1}{2}\Rea (\alpha) \nabla I
\end{equation}
due to a spatially varying ac Stark shift experienced by the cold atoms in a light field.  Two coherent off-resonance counter propagating lasers interfere with each other and form a standing wave, which leads to spatially modulated intensity with period of $\lambda/2$, where $\lambda$ is the wave length of the laser.  Depending on the sign of detuning $\Delta$, the atoms accumulate at the nodes (blue detuning) or the antinodes (red detuning)~\cite{CCTAdvAtomPhys}.    The quantized motion of the atoms comprise vibrational motion within an individual well and the tunneling between the neighbouring wells, leading to a band spectrum, whose dynamics can be described by a Boson-Hubbard Model~\cite{Jaksch813108}.   Apart from simulating the condensed matter physics~\cite{Jaksch813108}~\cite{FisherSF_INS30},  the cold atoms in optical lattice can also be used for quantum computation~\cite{Jaksch32Entanglement}.  Transition to the Mott insulator state was supposed to be an efficient way of preparing a quantum register with a fixed number of atoms per lattice site~\cite{LewensteinRev}.  This superfluid-Mott insulator transition was observed experimentally~\cite{BlochMOTTTransNature33}.

Even though the atoms in an optical lattice has close analog to the electrons in a crystal, there are several important differences.  The spatial order does not result from interactions between atoms but from an external potential created by light~\cite{CCTAdvAtomPhys}.  The distance of two atoms are quite large large (of wave length), however their interaction may be tuned by Feshbach resonance~\cite{CourteilleFeshbachResonance}~\cite{FeshbachInouye}.  The optical lattice is thus very flexible by tuning parameters of lasers and atom-light coupling.  The developments of these techniques enlarges the accessible range with cold atoms, enabling people to study strongly correlated system of complex quantum liquid~\cite{BlochDalibardRevModPhys.80.885}.   For review of  details about trapping atoms and forming optical lattice, the reader may resort to Ref.~\citenum{Grimm200095}.  For recent review concerning the developments and application of optical lattice in the field of quantum gases one may resort to Ref.~\citenum{LewensteinRev} and ~\citenum{BlochDalibardRevModPhys.80.885}.

\subsection{Optical Flux Lattice}
In analog to the crystal electron moving under external electromagnetic field, we would like to apply a ``magnetic field'' to our atoms in optical lattice.  As we have demonstrated in the previous sections, various schemes were proposed to generate an artificial gauge field in cold atom system.  However a typical scale of the flux density is of order $n_{\phi}\sim1/(L\lambda)$, where $L$ is a macroscopic length much larger than $\lambda$.  In order to reach the strongly correlated regime, one requires a filling factor $\nu\sim1$, i.e., $L$ should be as large as $\lambda$.  It is Cooper who proposed a possible scheme to generate such a strong ``magnetic field'' in optical lattice, which is referred as Optical Flux Lattice (OFL)~\cite{CooperOFL}.

In section~\ref{sec:TLA}, we have shown that the Hamiltonian of a two level atom moving in a light field is of form
\begin{equation}\label{eq:HamTLAtom}
  H= \frac{\mathbf{P}^2}{2M} \hat{\mathbb{I}} +U,
\end{equation}
where $U$ is a $2\times2$ Hermitian matrix dictating the atom-light coupling.  It can be written in general as
\begin{equation}\label{eq:HamALcoupling}
  U=V\mathbf{M}\cdot\hat{\bm{\sigma}}=
  V\begin{pmatrix}  M_z &  M_x-iM_y\\ M_x+iM_y & -M_z \end{pmatrix}.
\end{equation}
Following the same procedure of projecting to one subspace of eigenstate $|\chi\rangle$, we obtain the effective Hamiltonian of adiabatic orbital motion with an induced vector and scalar potential.  The vector potential is given by
\begin{equation}\label{eq:HamALcoupling}
  \mathbf{A}=i\hbar\langle\chi|\nabla|\chi\rangle,
\end{equation}
and the density of magnetic flux quanta is
\begin{equation}\label{eq:DoFluxQuant}
  n_{\phi}=\frac{\nabla \times\mathbf{A}}{\Phi_0}=\frac{i}{2\pi}\nabla\times\langle\chi|\nabla|\chi\rangle,
\end{equation}
where flux quanta $\Phi_0=h$ \footnote{Here we have taken $q^*=1$. The definition of the flux quantum differs from that defined in superconductor where $q^*=2q$.  Factor $2$ comes from the formation of cooper pair.}.

Now let us deviate to consider some geometrical and topological argument. In quantum Hall effect, we applied a magnetic field perpendicular to the system in the x-y plane.  This magnetic field is described by a vector potential, which explicitly breaks the translational symmetry (either with symmetrical gauge or Landau gauge).  Whereas with this magnetic field, electrons forms Landau levels that are topologically non-trivial.  However, Haldane~\cite{HaldaneModel} demonstrated that the topologically non-trivial band can be achieved without Landau levels preserving the translational symmetry.  The non-triviality is indicated by Chern number defined as integral of Berry curvature over the First Brillouin Zone (FBZ). The parameter space of the Haldane model is the quasimomentum vectors in FBZ with a periodic boundary, i.e. all the momentum vectors within FBZ forms torus topologically.  Mathematically, it means that integral of a closed form over a compact manifold without boundary, or i.e. closed form might not be exact, which reflects the non-trivial topology of the base manifold through de Rham cohomology.

With same reasoning, the vector gauge potential defined above is a connection of the frame bundle with basis $<|\chi\rangle>$, and the magnetic field is the curvature 2-form induced from $dA$ that is closed (only valid for Abelian case).  Both of them are defined periodically over the optical lattice in x-y plane, or equivalently over a base manifold of torus by periodical identification.  Different from the case by Haldane, the base manifold now is unit cell rather than torus.  The integral of the curvature form may give Chern number other than $0$ as long as we have non-trivial frame bundle $|\chi(\mathbf{R})\rangle$ (the base manifold of torus is topologically non-trivial).

More specifically, we, following Cooper's idea~\cite{CooperOFL}, define the local Bloch vector
\begin{equation}\label{eq:Blochvector}
  \mathbf{n}(\mathbf{r})=\langle \chi|\hat{\bm{\sigma}}|\chi\rangle,
\end{equation}
for which $\mathbf{n}\cdot\mathbf{n}=1$.
The flux density follows as
\begin{equation}\label{eq:FluxDensity}
  n_{\phi}=-\frac{1}{8\pi}\epsilon_{ijk}\epsilon_{\mu\nu}{n}_i(\partial_\mu{n}_j)(\partial_\nu{n}_k),
\end{equation}
where $n_i$ represents the components of Bloch vector $\mathbf{n}$, and $\partial_\mu$ is the partial derivative with respect to the local coordinates.  The antisymmetric tensor $\epsilon_{\mu\nu}$ arises from exterior differential operators, and $\epsilon_{ijk}$ is induced from volume form in three dimensional space (the spin operator $\hat{\bm{\sigma}}$ has three components).  The expression of the flux density has a natural geometric meaning in analog to the Gauss-Weingarten map.  As it is pointed out by Cooper~\cite{CooperOFL}, the flux through an area $A$ is given by $\int_A n_{\phi} (d^2\mathbf{r})=\Omega/4\pi$ where $\Omega$is the total solid angle that region $A$ maps to on the Bloch sphere. Therefore total flux over a unit cell counts the number that Bloch sphere were wrapped and gives an integer number $N_\phi$.

In Ref.~\citenum{CooperOFL}, Cooper projected $\mathbf{n}$  to the x-y plane and obtained a two dimensional vector field over the unit cell with periodic boundary condition.  He also noticed that the net flux has close relation to the singular points of the vector field (where the vector field vanishes)~\cite{CooperOFL}.  Here we would like to interpret these results in a more geometrical way, which may help us calculate the Chern number without heavy numerical calculus.

It is plausible to assume that the singularity of the projected vector field is isolated.  Thus we may consider a closed disc $D$ containing the singularity $p$ as an inner point, and define the index of $p$ as the degree of the map $f:\partial D\rightarrow S^1$.  $f$ maps the vector $\mathbf{v}=\Pi_\perp\mathbf{n}$ at a point of $\partial D$ to the unit sphere by $\mathbf{v}/|\mathbf{v}|$.  The degree of $f$ is simply winding number over $S^1$ as one turns  once around $D$.  Since $\mathbf{n}$ is a unit vector, at singularity $p$, $\mathbf{v}=\mathbf{0}$ means  $\mathbf{n}$ points either along positive z-direction or negative. Physically, This means that the off diagonal entries of the atom-light coupling vanish.  The examples given in Ref.~\cite{CooperOFL} all have singularities with $\Deg f=\pm1$.   All the interpretations following is easy to be generalized to singularities of higher indices which just means wrapping the sphere more than once.  So we shall focus ourself only on the simplest cases of  $\Ind (p)=\pm1$ without loss of generality.

Consider a singular point $p$ where $\mathbf{n}(p)=(0,0,1)$.
\begin{figure}
  \centering
    \subfigure[a][Bloch vector]{\includegraphics[width=0.35\textwidth]{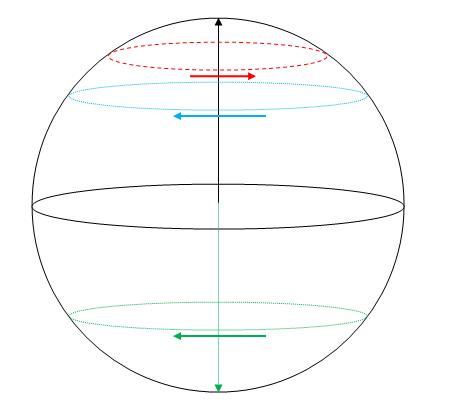}\label{fig:Bloch}}
    \subfigure[b][Projected vector field around a singularity]{\includegraphics[width=0.3\textwidth]{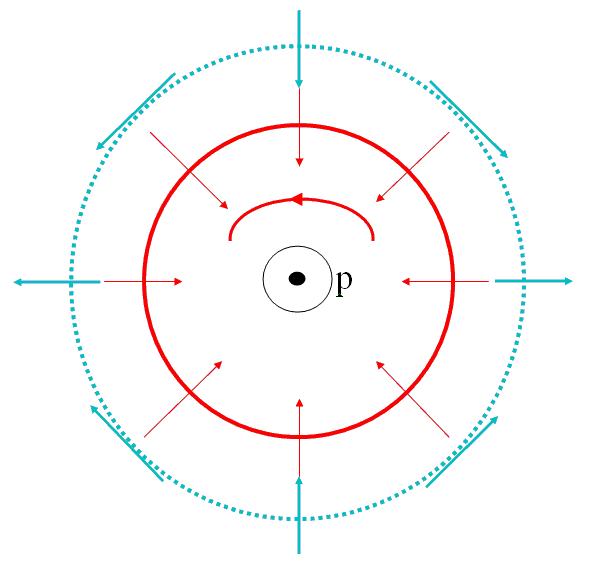}\label{fig:Hopf}}\\
  \caption{Index of Singularities}\label{fig:BlochHopf}
\end{figure}
The project vector field around $p$ can be either red case or the blue case shown in Fig.~\ref{fig:Hopf}.  If we turn around counter clockwisely (in the direction of red arrow) $p$ along the red contour, the field vector mapped to the Bloch sphere turns around the North pole along the red dash-line counter clockwisely (also in the direction of red arrow) shown in Fig.~\ref{fig:Bloch}.  So the singularity of this type (red case) is said to have an index $\Ind(p)=1$.  Similarly, in the case of blue, the field vector turns clockwisely on the Bloch sphere as one turn around singularity counter clockwisely, where the singularity is said to have an index $\Ind(p)=-1$.  (For the singularity with $\mathbf{n}$ pointing to South pole as shown by green arrow, the index can be defined in a similar way. To avoid ambiguity, we shall always turn around the singularity counter clockwisely from top view. If the vector under the Bloch map also truns counter clockwisely from the view of North pole, the singularity has index $1$; otherwise it has index$-1$.)  If we choose a disc $D$ small enough (it is always possible as singularity $p$ is isolated), the vectors will be mapped to the corresponding cap on the Bloch sphere, and the flux through that region is proportional to the solid angle of the cap.  If we choose a larger disc, the cap enclosed by the contour on the Bloch sphere may expand, shrink, tilt or even deformed.  However as long as the boundary of our domain in the x-y plane does not traverse any other singularity, everything will be well defined. One must be careful to the orientation of the integral here.  Since $\epsilon_{ijk}$ appears in Eq.~\ref{eq:FluxDensity}, the cap enclosing North pole with a counter clockwise boundary has a positive orientation, whereas the one enclosing South pole with a clockwise boundary (from the view of North pole) has positive orientation.

Now we enlarge the domains, each containing one the singularity, to cover the unit cell, but keep the boundaries of the domains not traversing other singularities nor crossing each other.  Then the integral over the unit cell is exactly the sum of the integral over every domain. It is worthy to note that all the boundaries cancel each other exactly in a sense that the  orientations should be taken into account. (The boundary of the unit cell was identified under periodic boundary condition.  It has a simple topological root  that torus has no boundary.)  If we consider the corresponding counter on the Bloch sphere, they should coincide with each other but with opposite orientation, which means the caps enclosed by them either have wrapped the Bloch sphere completely  for finite times or they cancel each other exactly. (The integral over same cap with opposite boundary orientation differs by a minus sign.)  In the former case, the solid angle wrapped by the caps is an integer multiples of $4\pi$ that means a net quantized flux (integer multiples of $\Phi_0$) through unit cell.

Another remarkable feature deserve to be emphasised here is that the singularities (with index $\pm1$) of the vector field over the  torus must appear pairwisely.  Due to Poincar\'e-Hopf index theorem
  \begin{equation}\label{eq:PoincareHopfthm}
    \sum_{All\  singularities\  p} \Ind(p)=\chi=2-2g,
  \end{equation}
where $\chi$ is Euler characteristic and $g$ is the genus of the manifold. In our case, torus has $g=1$, so $\chi=0$.  Thus the pair wise appearance of the singularities is the direct corollary from this theorem.  If there are only two singularities in the unit cell, and both of them has the vector $\mathbf{n}$ pointing to North pole, they must cancel themselves and no net flux penetrates the unit cell.  If they point in opposite direction, they should wrapping the sphere one time and do contribute to the net flux cooperatively.  So, with this feature, as long as the assumptions of isolated singularities is valid, we can read off Chern number directly from the projected vector field instead of tedious numerical computations.

In Ref.~\citenum{CooperOFL} pointed out that the singularities depends on the gauge chosen.  Mathematically, this corresponds that the choice of connection is not unique.  However, all the properties we concerned arise from the non-triviality of the frame bundle, which has a strong topological root that is independent of the specific choice of the connection.  Back to physics, this means that all the physical properties are gauge independent, which is well-known to the physicists.

Back from the long deviation, with the scheme of OFL, we are able to generate a much stronger artificial magnetic field, which enable study the more strongly correlated regime previously unattainable.  A scheme of OFL for two photon dressed states was proposed by Cooper and Dalibard~\cite{CooperDalibardOFL}. Possibility of reaching Fractional Quantum Hall (FQH) states with QFL was also proposed~\cite{FQHwithOFL}.

\section{Conclusion and Outlook}
In this article we reviewed the historical developments in artificial gauge fields and spin-orbit couplings in cold atom systems. we worked out examples carefully and connected physical and mathematical formalisms of same objects in different ways.  We gave intuitive and accessible physical and mathematical interpretations to optical flux lattice which reveals its deep connections between physics and topology.

The cold atom techniques have developed over past decays.  Artificial gauge field in cold atom systems is a new laboratory tool for us to study novel quantum phenomena not only in cold atoms, but also condensed matter physics, quantum information and even particle physics and cosmology.  New ideas and technologies are emerging which enable us to reach much farther and understand better of the complex quantum systems.

\section*{Acknowledgements}
The author would express sincere gratitude to J.~Dalibard, and N.~Goldman for their helpful discussions.  The author would thank J.-L.~Li and G.-C.~Li for their help for mathematics.  The author also thank S.~Nascimb\`ene for the details concerning the experiments.

\section{Appendix}
\subsection{Mathematical Formulae}
\begin{thm}\label{thm:sindelta}
  The limit of the measure $(\sin \alpha t)/\alpha$ satisfies
  \[
    \lim_{t\rightarrow+\infty}\frac{\sin \alpha t}{\alpha}=\pi\delta(\alpha).
  \]
\end{thm}
\begin{proof}
  $\forall \alpha \neq 0, \exists \epsilon>0$, such that $0\notin [\alpha-\epsilon,\alpha+\epsilon]$.

  By Riemann-Lebesgue lemma, we have $\forall f(x) \in C[a,b]$
  \[
   \lim_{t\rightarrow+\infty} \int_{\alpha-\epsilon}^{\alpha+\epsilon} {dx f(x) \frac{\sin x t}{x}}=0.
  \]

  For $\alpha=0$, Riemann integral over $[-\epsilon,+\epsilon]$, is not well-defined; instead we calculate the integral over $I=[-\epsilon,-\delta)\cup(+\delta,+\epsilon]$ and take the limit of $\delta \rightarrow 0$. Since f is continuous, by mean-value theorem, we have
  \[
  \begin{split}
    \lim_{t\rightarrow+\infty} \int_{I} {dx f(x) \frac{\sin x t}{x}}
    =&\lim_{t\rightarrow+\infty} f(\theta\delta)\int_{I} {dx \frac{\sin x t}{x}}\\
    \rightarrow&\lim_{t\rightarrow+\infty} f(0)\int_{\tilde{I}} {dx \frac{\sin x t}{x}}\\
    =& f(0)\int_{-\infty}^{+\infty} {dy \frac{\sin y}{y}}=\pi f(0)
  \end{split}
  \]
\end{proof}

\begin{thm}\label{thm:sinsqdelta}
  The limit of the measure $(\sin^2 \alpha t)/\alpha t$ satisfies
  \[
    \lim_{t\rightarrow+\infty}\frac{\sin^2 \alpha t}{\alpha^2 t}=\pi\delta(\alpha).
  \]
\end{thm}
\begin{proof}
  Here we shall not give a detail proof.  But one may easily observe that for $\alpha\neq0$, the measure tends to $0$. Around $\alpha=0$, the measure diverges as $t$.  The constant can be obtained by residue formula if one write $\sin^2 x=(1-\cos 2x)/2$ as $(1-e^{2ix})/2$ and calculate a similar integral to the one in Theorem~\ref{thm:sindelta}.
\end{proof}


\begin{prop}\label{thm:diffs}
  $\langle d f|\wedge |d f \rangle =d \langle  f|d f \rangle$
\end{prop}
\begin{proof}
  Assume that internal degree of freedom is finite, the $| f \rangle$ can be written componentwisely in a specific basis $f={(f_1,f_2,...,f_N)}^T$.  And the inner product is
  \begin{equation*}
    \langle f|g\rangle = \sum_{i=1}^{N}f_i^*g_i.
  \end{equation*}

  Since $df={(df_1,df_2,...,df_N)}^T$,
  \begin{equation*}
     \begin{split}
       d \langle  f|d f \rangle
       =& d \left( \sum_{i=1}^{N}f_i^*df_i \right)  \\
       =&  \left( \sum_{i=1}^{N} d f_i^*\wedge df_i \right)
       = \langle d f|\wedge |d f \rangle,
     \end{split}
  \end{equation*}
  where linearity of the inner product was used in the first step (thus $d $ commutes with $\sum$), and $d^2=0$ was used in the second step.

\end{proof}

\begin{prop}\label{thm:nonabtrans}
One kind of intuitive choice of eigenvectors of ${\bm{\sigma}}\cdot{\bf n}$, where $\mathbf{n}=(\sin \theta \cos \phi,\sin \theta \sin \phi,\cos \theta)$ is by first choosing the eigenvectors of $\sigma_z$ as $|+\rangle={(1,0)}^T$ and $|-\rangle={(0,1)}^T$, and then rotating in y- and z- direction consecutively (z after y).  Therefore one obtain the eigenvectors in ${\bf n}$- direction is
\begin{equation}\label{eq:eigvecn}
  \begin{split}
  |+(\theta,\phi)\rangle =&
    \begin{pmatrix}
      e^{-i\phi/2} \cos (\theta/2)\\
      e^{i\phi/2}  \sin (\theta/2)
    \end{pmatrix}, \\
  |-(\theta,\phi)\rangle =&
    \begin{pmatrix}
      -e^{-i\phi/2} \sin (\theta/2)\\
      e^{i\phi/2}   \cos (\theta/2)
    \end{pmatrix},
  \end{split}
\end{equation}
which is different from the choice in Eq.~\ref{eq:inteignvect} up to a phase factor.  Thus the gauge potential induced by them is connected by the formula $\mathbf{A}'=U^\dagger\mathbf{A}U+i\hbar U^\dagger dU$, where $U$ transform one basis to another.\\
\end{prop}
\begin{proof}
In Eq.~\ref{eq:inteignvect} we have used base vectors
\begin{equation*}
  \begin{split}
  |\chi_1\rangle =&
    \begin{pmatrix}
      \cos (\theta/2)\\
      e^{i\phi} \sin (\theta/2)
    \end{pmatrix}, \\
  |\chi_2\rangle =&
    \begin{pmatrix}
      -e^{-i\phi} \sin (\theta/2)\\
      \cos (\theta/2)
    \end{pmatrix}.
  \end{split}
\end{equation*}
They are transformed to our new basis $|\pm\rangle$ by a unitary transformation
\begin{equation*}
  \begin{pmatrix}|+\rangle & |-\rangle \end{pmatrix}=
  \begin{pmatrix}|\chi_1\rangle & |\chi_2 \rangle \end{pmatrix}
  \begin{pmatrix}
      e^{-i\phi/2} &   0\\
      0            &  e^{i\phi/2}
    \end{pmatrix}
    =\begin{pmatrix}|\chi_1\rangle & |\chi_2 \rangle \end{pmatrix} U.
\end{equation*}
\begin{equation*}
   \begin{split}
     i\hbar U^\dagger dU=& i\hbar
  \begin{pmatrix}
      e^{i\phi/2} &   0\\
      0            &  e^{-i\phi/2}
    \end{pmatrix}
    \begin{pmatrix}
      \frac{-i}{2}d\phi e^{-i\phi/2} &   0\\
      0            &  \frac{i}{2}d\phi e^{i\phi/2}
    \end{pmatrix}\\
    =&\begin{pmatrix}
      \frac{\hbar }{2}d\phi &   0\\
      0            &  -\frac{\hbar}{2}d\phi
    \end{pmatrix}
   \end{split}
\end{equation*}
\begin{equation*}
\begin{split}
  {\mathbf{A}}'=& U^\dagger {\mathbf{A}} U + i\hbar U^\dagger dU\\
  =&
\begin{pmatrix}
  \hbar [\frac{1}{2}-\sin^2 \left(\frac{\theta}{2}\right)] d \phi &
  -\frac{\hbar}{2}(id \theta +(\sin\theta )d \phi)\\
  \frac{\hbar}{2} (id \theta -(\sin\theta )d \phi)
  & \hbar [\sin^2 \left(\frac{\theta}{2}\right)-\frac{1}{2}] d \phi
\end{pmatrix}.
\end{split}
\end{equation*}
\end{proof}

\bibliographystyle{ieeetr}
\bibliography{artificialgaugeandsoc}

\end{document}